%% file: ICDE2020HABF.tex
\documentclass[10pt, conference, letterpaper]{IEEEtran}

\usepackage{trackchanges}
\usepackage{booktabs}
\usepackage{wrapfig,graphicx}
\usepackage{lipsum}
\usepackage{multirow}
\usepackage[ruled, linesnumbered, algo2e]{algorithm2e}
\definecolor{modified}{RGB}{56, 103, 255}
\soulregister\ref7
\soulregister\cite7
\soulregister\eg7
\input{new-commands}

\newcommand{\RNum}[1]{\uppercase\expandafter{\romannumeral #1\relax}}

\addeditor{lm}

\begin{document}

\title{Hash Adaptive Bloom Filter
\vspace{-1mm}
}

\author{\IEEEauthorblockN{
Rongbiao Xie$^{1\dagger}$~~
Meng Li$^{1\dagger}$~~
Zheyu Miao$^{2}$~~
Rong Gu$^{1*}$~~
He Huang$^3$~~
Haipeng Dai$^{1*}$~~
Guihai Chen$^1$
}
\IEEEauthorblockA{State Key Laboratory for Novel Software Technology, Nanjing University, Nanjing, Jiangsu 210023, CHINA$^1$, \\
Zhejiang University, Hangzhou, Zhejiang 310058, CHINA$^2$,\\
School of Computer Science and Technology, Soochow University, Suzhou, Jiangsu 215006, CHINA$^3$}
\{rongbiaoxie, menson\}@smail.nju.edu.cn, \{gurong, haipengdai, gchen\}@nju.edu.cn, \\ predatory@zju.edu.cn, huangh@suda.edu.cn
\vspace{4mm}
}

\renewcommand{\thefootnote}{\fnsymbol{footnote}}
\maketitle
\newcommand\blfootnote[1]{%
\begingroup
\renewcommand\thefootnote{}\footnote{#1}%
\addtocounter{footnote}{-1}%
\endgroup
}

\blfootnote{$\dagger$ R. Xie and M. Li are the co-first authors. $*$ R. Gu and H. Dai are the corresponding authors.}
{\sloppy
\input{Latex/abstract.tex}
%\IEEEpeerreviewmaketitle
\input{Latex/introduction.tex}
\input{Latex/related.tex}
\input{Latex/model.tex}
\input{Latex/analysis.tex}

\input{Latex/evaluation.tex}

\input{Latex/conclusion.tex}
\input{Latex/acknowledge.tex}
}

%\vspace{-0.8em}
%\input{INFOCOM2019ROBUSTPERSISTENT.bib}
\bibliographystyle{IEEEtran}%---

\bibliography{ICDE2020HABF.bib}
\end{document}

%% file: new-commands.tex
\mathchardef\Gamma="0100 \mathchardef\Delta="0101
\mathchardef\Theta="0102 \mathchardef\Lambda="0103
\mathchardef\Xi="0104 \mathchardef\Pi="0105
\mathchardef\Sigma="0106 \mathchardef\Upsilon="0107
\mathchardef\Phi="0108 \mathchardef\Psi="0109
\mathchardef\Omega="010A

\newcommand{\outline}[1]{}%{\textbf{#1}}

\usepackage{cite}
\usepackage{xspace}
\usepackage{url}
\usepackage{graphicx}
\usepackage{latexsym}
\usepackage{amssymb}
\usepackage{amsfonts}
\usepackage{psfrag}
\usepackage{wrapfig}

\usepackage{algorithm}
\usepackage{algorithmicx}
\usepackage{algpseudocode}
\usepackage{alltt}
\usepackage{color}
\usepackage{array}
\usepackage{mdwmath}
\usepackage{mdwtab}
\usepackage[tight,footnotesize]{subfigure}
\usepackage{amsfonts}

\usepackage{amssymb}
\usepackage{tabularx}
\usepackage{mathrsfs}
\usepackage{pdfsync}
\usepackage{amsmath,bm}
\usepackage{caption2}
\usepackage{cite}
\usepackage{textcomp}

\newcommand{\ie}{\emph{i.e.}\xspace}
\newcommand{\eg}{\emph{e.g.}\xspace}

\newcommand{\etal}{\frenchspacing{}\emph{et al{.}}\xspace}

\newtheorem{lemma}{\textbf{Lemma}}[section]

\newtheorem{theorem}{\textbf{Theorem}}[section]

\algblock{ParFor}{EndParFor}
\algnewcommand\algorithmicparfor{\textbf{parfor}}
\algnewcommand\algorithmicpardo{\textbf{do}}
\algnewcommand\algorithmicendparfor{\textbf{end\ parfor}}
\algrenewtext{ParFor}[1]{\algorithmicparfor\ #1\ \algorithmicpardo}
\algrenewtext{EndParFor}{\algorithmicendparfor}

%% file: Latex/abstract.tex
\vspace{-6mm}
\begin{abstract}
Bloom filter is a compact memory-efficient probabilistic data structure supporting membership testing, \ie, to check whether an element is in a given set.
%with a hash table as the underlying data structure.
%
However, as Bloom filter maps each element with uniformly random hash functions, few flexibilities are provided even if the information of negative keys (elements are not in the set) are available.
The problem gets worse when the misidentification of negative keys brings different costs.
To address the above problems, we propose a new \underbar{H}ash \underbar{A}daptive \underbar{B}loom \underbar{F}ilter (HABF) that supports the customization of hash functions for keys.
The key idea of HABF is to customize the hash functions for positive keys (elements are in the set) to avoid negative keys with high cost, and pack customized hash functions into a lightweight data structure named HashExpressor.
Then, given an element at query time, HABF follows a two-round pattern to check whether the element is in the set.
Further, we theoretically analyze the performance of HABF and bound the expected false positive rate.
We conduct extensive experiments on representative datasets, and the results show that HABF outperforms the standard Bloom filter and its cutting-edge variants on the whole in terms of accuracy, construction time, query time, and memory space consumption (Note that source codes are available in \cite{ourSourceCode}).
\end{abstract}

%% file: Latex/introduction.tex
%\vspace{-1mm}
\section{Introduction}
\label{sec:intro}
Membership testing problem refers to testing whether an item is in a given set.
It is a fundamental problem in numerous applications such as big data applications and databases, where the query latency, memory consumption, and accuracy are the primary performance indicators.
To address the problem, a lightweight probabilistic data structure named Bloom filter, with a bit vector of length $m$ as the underlying data structure, is proposed \cite{Bloom1970Space}.
To insert an item into Bloom filter, the item is mapped by $k$ hash functions to $k$ bits in the bit vector, and all $k$ mapped bits are set to $1$.
To query an item, the item is mapped by $k$ hash functions, and it is considered to be a member if all $k$ mapped bits are $1$.
Due to the compact space efficiency and satisfactory accuracy, Bloom filter has been the common practice in many applications.
For example, it is used to avoid unnecessary I/O overhead \cite{sears2012blsm} base on log-structured merge (LSM) tree \cite{LSMtree} (\eg, LevelDB \cite{LevelDB} and RocksDB \cite{RocksDB}) in some key-value databases;
to reduce communication cost in a distributed database \cite{mackert1994r*}, 
and to prevent Distributed Denial-of-Service (DDoS) in network security \cite{xiao2006novel}.
However, there is a small probability for Bloom filter to mistakenly identify a negative key (\ie, a key is not in the set) as a positive key (\ie, a key is in the set), which is called false positive~\cite{graf2020xor}.
Targeting at decreasing the number of false positives, a lot of techniques~\cite{graf2020xor, cohen2003spectral, guo2009dynamic, kirsch2006less, hao2007building, deng2006approximately, mitzenmacher2002compressed} have been developed over the past decades.
However, these works shared a similar technical path, \ie, reducing the false positives by leveraging the randomness (in terms of hash function mapping) while ignoring the availability of negative keys in many systems~\cite{kraska2018the,mitzenmacher2018a, 2019adaptive, rae2019meta, bhattacharya2020adaptive}.
For example, for intrusion detection, malicious IP address statistics can be obtained from access logs or some well-known online real-time Blacklists such as URIBL~\cite{URIBL};
and for LSM-tree-based key-value databases~\cite{LevelDB, RocksDB}, the frequently failed queries with heavy I/O overhead can be cached to reduce extra disk accesses.
Unfortunately, such negative key information is hardly utilized~\cite{graf2020xor, cohen2003spectral, guo2009dynamic, kirsch2006less, hao2007building, deng2006approximately, mitzenmacher2002compressed}.

Recently, Bloom filters empowered by machine learning (ML) techniques~\cite{kraska2018the, mitzenmacher2018a, 2019adaptive, rae2019meta, bhattacharya2020adaptive} are proposed to take advantage of the keys information (including negative keys) by introducing in learned models.
However, they suffer from the explosive growth of latency, for query and insert operations (\eg, $400\times$ of standard Bloom filter~\cite{rae2019meta}), which is quite intolerant for many latency-sensitive applications.
Not to mention the extra computation overhead incurred by the time-consuming training phase of learned models, which although can be alleviated with hardware (\eg, GPU and TPU).
Therefore, how to build a practical filter that takes advantage of the negative keys information remains unknown.
Besides, we consider one further problem, that is, in many real-world systems, the misidentification of negative keys (\ie, false positives) brings cost, which may even be highly skewed from one key to another~\cite{babcock2003distributed, cormode2005s, wu2020ac}.
For example, Internet traffics is highly skewed and concentrates on some popular files \cite{breslau1999web}, and popular files will bring more communication costs than unpopular files.
In LevelDB, accessing data in different levels incurs significantly different I/O costs from disk accessing \cite{zhang2018elasticbf}.
Besides, some cost information can be or is already being monitored~\cite{babcock2003distributed, cormode2005s, wu2020ac, breslau1999web, zhang2018elasticbf, 2006Weighted}.
However, Bloom filter cannot directly utilize such cost information because it treats all keys equally by sharing identical fixed $k$ hash functions.
Note that the situation may even get worse if the shared hash functions are not uniformly random or even skewed.

Before diving into our proposed solutions, we formally define the problem considered in this paper as follows:
%
%\vspace{-1mm} \textbf{Problem Definition.}
%
\textit{Suppose the positive key set is denoted by $S$, the negative key set is denoted by $O$, the global hash function set is $H$, the number of hash functions is $k$, and for a certain key $e$, the cost of $e$ is $\Theta(e)$.
Our problem is how to build a Bloom filter so that the overall cost of false positives from $O$ is minimized?}
To address the problem, we propose a new solution, \ie, customizing hash functions for each key in Bloom filter individually according to the given negative keys and their costs information.
To be specific, we aim to select a hash function set of size $k$ for each key from $H$ to construct a Bloom filter so that the overall cost of false positives from $O$ is minimized.
Besides, the hash customization mechanism avoids the performance degradation from hash function skewness.

\begin{figure}[t]
	\vspace{2mm}
	\centering
	\setlength{\abovecaptionskip}{5pt}
	\includegraphics[width=0.7\linewidth]{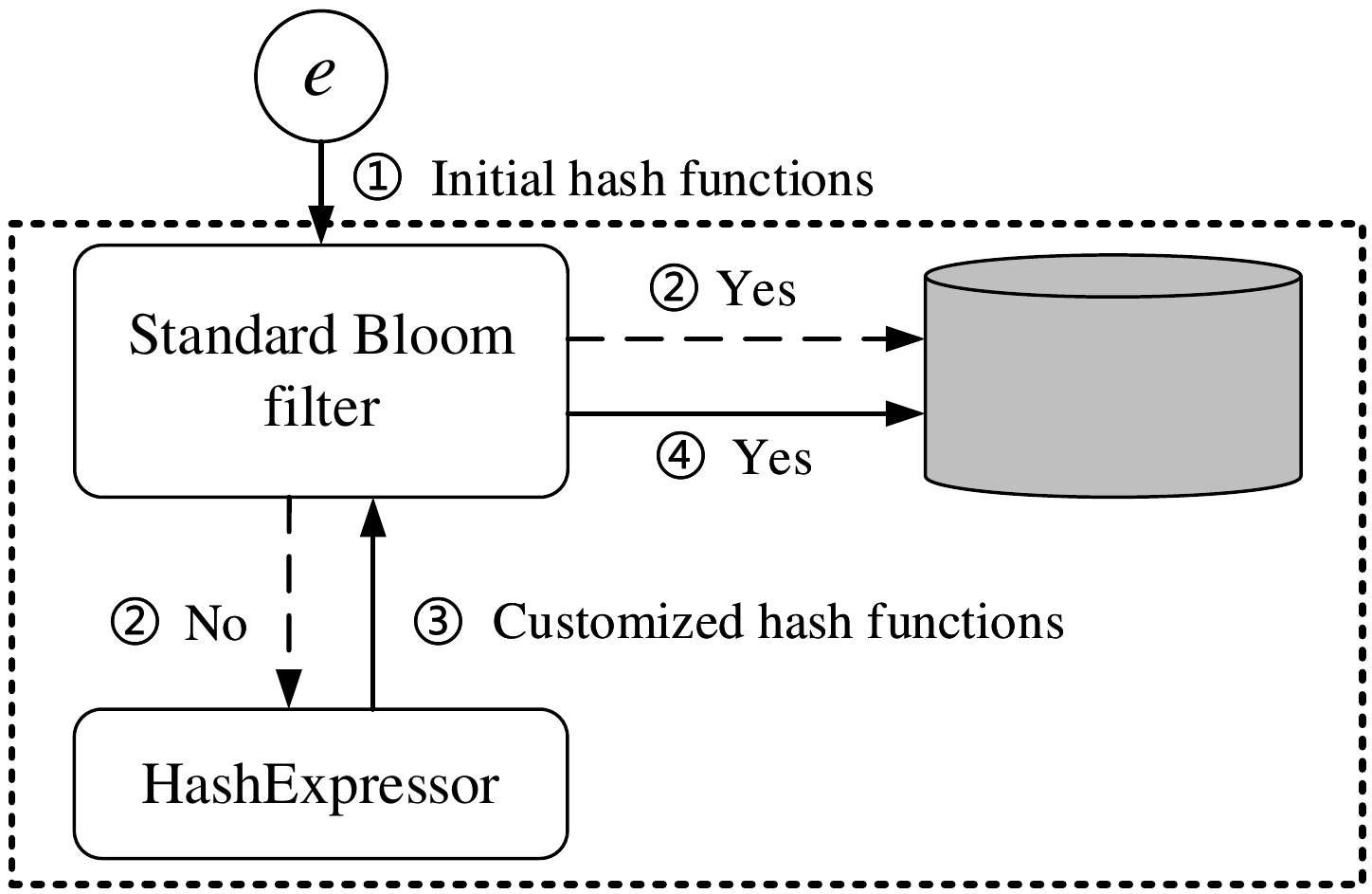}
	\vspace{0mm}
	\caption{Architecture of HABF}
	\label{fig:HABFstructure}
	\vspace{-10mm}
\end{figure}

In this paper, we propose a novel structure named Hash Adaptive Bloom Filter (HABF), which consists of two parts including a standard Bloom filter and a novel lightweight hash table named HashExpressor, as shown in Fig. \ref{fig:HABFstructure}.
The key idea of HABF is to customize and store the hash functions at construction time and then obtain customized hash functions at query time.
Note that storing hash functions for each key is non-trivial as it consumes large memory.
Instead, HABF allocates \emph{initial hash functions} for each key, and then adjusts the hash functions for only a small portion of positive keys that incur false positives.
%
%%Construction
During the construction time, we first allocate all keys with $k$ random \emph{initial hash functions} from the global hash functions collection that is available in Table \ref{Tab:hashfunc}, and then optimize hash function selections for positive keys with our proposed \underbar{T}wo-\underbar{P}hase \underbar{J}oint \underbar{O}ptimization (TPJO) algorithm, which is greedy-based but with performance bound.
After obtaining the optimal hash function selections, we pack them into the previously mentioned lightweight hash table HashExpressor.
%
%%Query
During the query time, a key $e$ first applies \emph{initial hash functions} to check whether it is positive.
If yes, $e$ is believed to be positive.
If no, we query a new set of hash functions from HashExpressor and check with Bloom filter again.
If yes, $e$ is also believed to be positive; otherwise, $e$ is considered to be negative.
Following such a two-round pattern, HABF has no false negatives as the standard Bloom filter.

As far as we know, there is no prior work on customizing hash functions for keys to address the defined problem above.
The most related work is to group keys into disjoint subsets and use a different hash function set for each subset \cite{hao2007building} to decrease the number of bits equals $1$ and optimize false positive rate (FPR).
In a sense, it is only a special case of customizing hash functions.
Considering that HABF customizes hash functions according to negative keys and their cost, the information needs to be known during construction time.

\textbf{Challenges.} In this paper, we are mainly faced with three challenges.
The first challenge is how to customize hash functions for positive keys to minimize the overall cost of our problem, as a brute-force search brings exponential complexity.
To address the challenge, we propose a performance-bounded greedy-based algorithm named TPJO to find the optimal hash functions.
The second challenge is how to store the optimal hash functions of adjusted keys without incurring heavy space overhead.
To address the challenge, we design a hash table named HashExpressor by sharing the same space.
The third challenge is how to ensure that HABF inherits the nice query performance of the standard Bloom filter, \ie, no false negative rate (FNR) and a small FPR.
Considering that each key in HABF is mapped with \emph{initial hash functions} or customized hash functions in HashExpressor, to address this challenge, the query of HABF follows a two-round pattern.
A key is negative if and only if it is checked not to be in the set after the two-round query.
\textbf{Contributions.} Our principal contributions can be summarized as follows.
Firstly, we consider the scenarios where the information of negative keys and cost can be obtained, and we propose a novel framework named HABF.
Secondly, we theoretically analyze the performance of HABF and bound the expected false positive rate.
Thirdly, we evaluate the proposed framework on representative datasets to validate its effectiveness and efficiency.
The results show that our HABF achieves high accuracy and low cost under the scenarios that the negative keys and their costs information can be obtained when using the same space size.
The rest of this paper is organized as follows.
We first review related works in Section \ref{sec:related}.
Then we present the architecture of HABF together with the construction/query procedure in Section \ref{sec:model}.
Next, we give the theoretical analysis in Section \ref{sec:analysis}.
After that, we present our experimental result in Section \ref{sec:eval}.
Finally, we conclude our work in Section \ref{sec:conclusion}.

%% file: Latex/related.tex
\vspace{2mm}
\section{Related Work}
%\vspace{-2mm}
\label{sec:related}
In this section, we first review the standard Bloom filter \cite{Bloom1970Space}, and then three types of variants closely related to our work.

\textbf{Bloom filter.}
The standard Bloom filter \cite{Bloom1970Space} has a bit array as the underlying data structure, and supports membership testing query.
Bloom filter provides a one-side error guarantee, \ie, small FPR and zero FNR.
To be specific, if a key is indicated to be absent in the set by the query result, it is definitely not in the set (zero FNR).
In contrast, if the key is indicated to be in the set, it is actually not in the set with a small error probability (FPR).
Given the number of bits allocated for each key (bits-per-key) $b$, the FPR can be formulated as $(1-e^{-\frac{k}{b}})^k$\cite{kirsch2006less} and achieves its minimum value of $0.6185^b$ when $k=ln2 \cdot b$.
Unfortunately, as Bloom filter shares $k$ identical hash functions across all keys, it is insensitive to the information of negative keys and cost.

\textbf{Hash function/fingerprint-based.}
Gosselin-Lavigne \etal evaluated different hash functions and selected several optimal ones in terms of FPR as the default functions for Bloom filter \cite{2015A}.
However, they only aimed at seeking hash functions with better implementations.
Hao \etal proposed to group keys into disjoint subsets and used a different set of hash functions for each subset \cite{hao2007building}.
In contrast, we can achieve fine-grained hash functions customization for each key.
For static datasets, Broder \etal proposed to store a fingerprint of each key in its corresponding hash location \cite{broder2004network} by designing a \emph{perfect hash function} to achieve optimal memory usage.
The fingerprint is generated by a hash function, and a key is considered to be positive only when its fingerprint is matched.
Nonetheless, the construction incurs heavy computation overhead.
Recently, a new filter named Xor filter~\cite{graf2020xor} is proposed with optimal memory usage.
However, no further performance gain is achieved by it when negative keys and costs are known.

\textbf{Cost-based.}
Considering the cost of different keys, Bruck \etal proposed Weighted Bloom filter (WBF) to reduce the overall cost by setting the number of hash functions for each key according to its cost.\cite{2006Weighted}.
However, when it comes to the query phase, it relies heavily on the cost to calculate the number of used hash functions for each key, which accordingly incurs large additional memory consumption and high query latency from storing and retrieving cost information.
Zhong \etal also studied how to adjust the number of hash functions based on cost of keys, and posed it as a constrained nonlinear integer programming problem together with two polynomial-time approximation solutions\cite{zhong2008optimizing}.
Similarly, this method incurs heavy space overhead to store the optimized number of hash functions, and high query latency when retrieving them.
%
%Cohen \etal proposed Spectral Bloom filter to extend the standard Bloom filter to support query-frequency-based multi-set membership testing, which filters elements with multiplicities below a given threshold\cite{cohen2003spectral}.
%
ElasticBF considers a different case where data (key and value) are stored in the multi-level LSM tree \cite{LSMtree}.
To relieve the I/O cost brought by accessing hot data in different levels, ElasticBF proposes to construct multiple small Bloom filters for each level and dynamically load the filter into memory as needed to achieve a fine-grained and elastic control on memory usage \cite{zhang2018elasticbf}.
However, ElasticBF only aims at cutting down I/O cost rather than the overall cost brought by FPR of Bloom filter.
%
%For multiple Bloom filters, Monkey \etal proposed to allocate different space to filters according to  the distributions of data \cite{dayan2017monkey}.

%%
\textbf{Learning-based.}
Kraska \etal first proposed Learned Bloom filter (LBF) by incorporating a machine learning (ML) model to improve space utilization from evident characteristics of data distribution \cite{kraska2018the}.
Mitzenmacher proposed to add an initial Bloom filter before ML model to improve the performance of LBF, which named Sandwiched Learned Bloom filter (SLBF) \cite{mitzenmacher2018a}.
Dai \etal proposed Adaptive Learned Bloom filter (Ada-BF) to score keys by ML model and tune the number of hash functions according to the score \cite{2019adaptive}.
Under incremental workloads, Bhattacharya \etal proposed two variants of LBF for supporting updates \cite{bhattacharya2020adaptive}, one is Classifier-Adaptive LBF (CA-LBF) by retraining ML model, the other is  Index-Adaptive LBF (IA-LBF) by sacrificing memory.
With an elaborately trained learned model, existing learning-based works could achieve remarkable performance in terms of FPR but at the cost of prolonged training time and query latency.
Besides, they are not sensitive to cost distribution.

%% file: Latex/model.tex
\vspace{2mm}
\section{Hash Adaptive Bloom filter}
\label{sec:model}
\vspace{1mm}
In this section, we first present the model of customizing hash functions for each key and formulate the optimization problem.
Then we provide the problem observation and our design insight.
Next, we describe the architecture of HABF in detail and the TPJO algorithm is further proposed to optimize the hash function selections.
Finally, the Zero-FNR query procedure is provided, followed by the FPR analysis.

%Finally, we give two discussions for HABF: 1) Fast construction and query; 2) Working under dynamic workloads;
%%
\subsection{Problem Formulation}
Let $U$ denote the universal key set.
Meanwhile, $S$ is a collection of positive keys in $U$ and $O$ is a collection of negative keys in $U$.
Note that $S$ and $O$ are disjoint.
Let $\Theta$ denote the cost distribution of keys, \ie, $\Theta(e)$ is the cost of key $e$.
Given the set of global hash functions $H=\{h_1, h_2,\cdots, h_{|H|}\}$.
Our problem is how to select a hash function subset $\phi(e)$ of size $k$ from $H$ for each key to minimize the overall cost brought by false positives of keys from $O$. To measure the performance across different algorithms, we define the normalized cost from false positives as a weighted FPR, namely,
\begin{align}\label{model:equa1}
   {\small Weighted \; FPR = \frac{\sum \limits_{e \in O}\Theta(e)\cdot \prod \limits_{h \in \phi(e)}\sigma(h(e))}{\sum \limits_{e\in O} \Theta(e)}},
\end{align}
where $\sigma(i)$ is the value of $i^{th}$ bit in Bloom filter.
In particular, when $\Theta$ is uniform, the weighted FPR is equivalent to traditional FPR.
For quick reference, we summarize the notations used throughout this paper in Table \ref{Tab:notation}.

\begin{table}[t]
	\vspace{0mm}
	\caption{Notations}
	\vspace{-2mm}
	\label{Tab:notation}
    \small
	\begin{center}
		\begin{tabular}{|c|p{6.3cm}|}
			\hline
			\text{Notations} &\text{Definitions}\\
			\hline
			$S, O$ & Collection of positive keys, negative keys in $U$\\
			\hline
            $e_s, e_o $ & Key in $S$, $O$\\
			\hline
            $m$ & Number of bits in Bloom filter\\
			\hline	
            $k$ & Number of hash functions used by keys\\
			\hline
            $\Theta(e)$ & Cost of key $e$\\
			\hline
            $\phi(e)$ & $k$-size hash function subset selected from $H$ for key $e$\\
			\hline
            $H$ & Global hash functions, $H=\{h_1, h_2,…, h_{|H|}\}$\\
			\hline
            $H_0$ & Initial hash function selection\\
			\hline
            $\omega$ & Number of cells in HashExpressor\\
			\hline
            $C[i]$ & $i^{th}$ cell in HashExpressor\\
			\hline
            $f$ & Unified hash function of HashExpressor\\
			\hline
            $e_{ck}, e_{opk}$ & collision key, optimized key\\
			\hline
            $V, \Gamma$ & Two runtime-index structures\\
			\hline
            $V[i], \Gamma[i]$ & $i^{th}$ unit in $V$, $i^{th}$ bucket in $\Gamma$\\
			\hline
            $F_{bf}, F^*_{bf}$ & False positive rate of Bloom filter in HABF before and after optimization \\
			\hline
		\end{tabular}
	\end{center}
\vspace{-10mm}
\end{table}

\subsection{Observation and Design Insight}
To optimize Equation~(\ref{model:equa1}), a straight design is to go through all possible hash function subsets for each key, and choose the one with the optimal weighted FPR.
However, such a brute-force method is time-consuming and incurs heavy space overhead, \ie, storing hash functions for each key.
Besides, we may use machine learning (ML) models to approximate and store the optimal hash function subset for each key, while the ML model needs to be elaborately trained and heavy computation overhead for training will inevitably be introduced.
Therefore, these designs are impractical.

Further, we observe that if the hash function subset of each negative key is fixed, the weighted FPR is only determined by the bits equal $1$, which are set by (inserted) positive keys.
Inspired by this, we randomly choose a set of hash functions as the \emph{initial hash functions} from $H$ for each (positive/negative) key and then adjust hash functions for certain positive keys to prevent them from conflicting with negative keys.
Therefore, the majority of (unadjusted) keys stick to the \emph{initial hash functions} while the (adjusted) positive keys switch to new hash functions.
Thus, we only need to store the hash functions of (adjusted) keys, rather than that of the universal keys.
Let $H_0=\{h_1^0, \cdots, h_k^0\}$ denote \emph{initial hash functions}.
\subsection{Architecture}
As shown in Fig.~\ref{fig:HABFstructure}, HABF consists of a standard Bloom filter and a data structure named HashExpressor.
At construction time, HABF customizes the hash functions for each positive key to reduce weighted FPR and stores the customized hash functions into HashExpressor.
At query time, it follows a two-round pattern by using $H_0$ first, and if the query with $H_0$ fails, then using hash function subset retrieved from HashExpressor.

As shown in Fig.~\ref{fig:hashexpressor}(a), HashExpressor is a probabilistic structure composed of $\omega$ cells, each of which is a $2$-tuple: $\left \langle endbit, hashindex\right \rangle$.
The $endbit$ field indicates whether the queried hash function subset comes from an adjusted positive key.
The $hashindex$ field stores the index of a hash function from $H$.
Let $C[i]$ be the $i^{th}$ cell of HashExpressor, $C[i].endbit$ and $C[i].hashindex$ be the $endbit$ and $hashindex$ of $C[i]$, respectively.
Here, $C[i]$ is empty if both $C[i].endbit$ and $C[i].hashindex$ are zero.
Now, we introduce the two basic operations of HashExpressor, \ie, Insertion and Query.

\textbf{1) Insertion.}
For each key $e$ and its hash function subset $\phi(e)$, we firstly initialize all hash functions in $\phi(e)$ to be invalid (not being inserted already).
Then, HashExpressor maps $e$ to the cell $C[f(e)]$ with a predefined hash function $f$, and there are three cases for cell $C[f(e)]$:

\emph{Case 1:} if $C[f(e)]$ is empty, we randomly choose an invalid hash function $h$ from $\phi(e)$ and mark $h$ as valid.

\emph{Case 2:} if $C[f(e)]$ is not empty and $C[f(e)].hashindex$ is an invalid hash function in $\phi(e)$, we mark $h = C[f(e)].hashindex$ in $\phi(e)$ as valid.

\emph{Case 3:} $\phi(e)$ is failed to be inserted.

\begin{figure}[t]
	\vspace{0mm}
	\centering
	\setlength{\abovecaptionskip}{5pt}
	\includegraphics[width=0.82\linewidth]{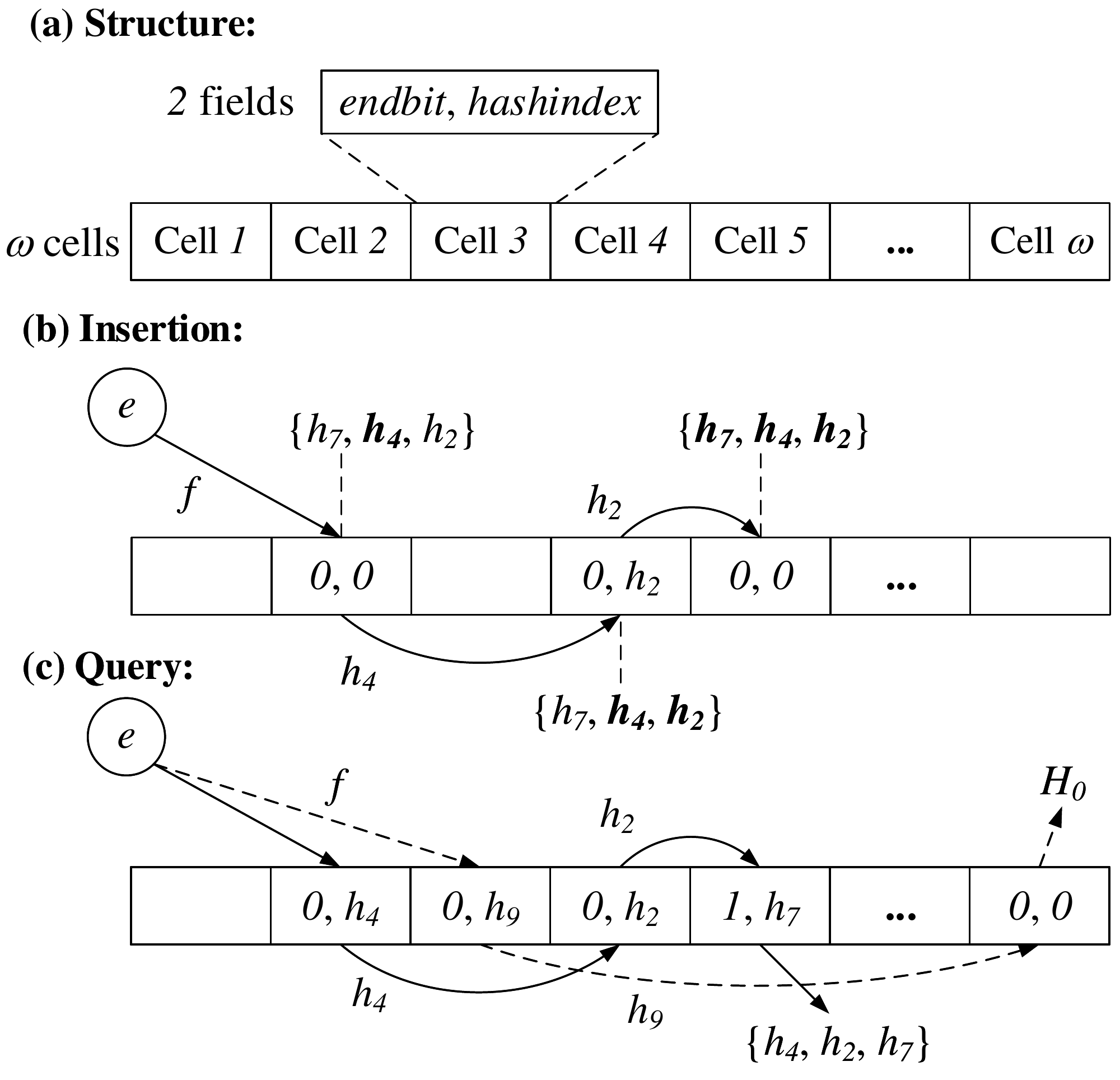}
	\vspace{-0.5mm}
	\caption{HashExpressor structure and operations}
	\label{fig:hashexpressor}
	\vspace{-10.5mm}
\end{figure}

If $C[f(e)]$ falls into Case 1 or 2, we repeat the above mapping procedure but with another hash function $h$, \ie, mapping $e$ to the next Cell $C[h(e)]$.
The above procedure repeats until all hash functions in $\phi(e)$ are marked as valid and the $endbit$ of cell mapped in the last time will be set to $1$. Then, we insert the hash functions in $\phi(e)$ into HashExpressor in the order of marking valid.
For example, as shown in Fig. \ref{fig:hashexpressor}(b), when inserting $\phi(e)=\{h_7,h_4,h_2\}$, $e$ is first mapped to an empty cell $\left \langle 0, 0\right \rangle$ with $f$, we randomly mark $h_4$ as valid.
Next, $e$ is mapped to cell $\left \langle 0, h_2\right \rangle$ with $h_4$, thus we mark $h_2$ as valid.
At last, $e$ is mapped to another empty cell $\left \langle 0, 0\right \rangle$ with $h_2$, we mark $h_7$ as valid and set the $endbit$ to $1$.
Finally, we insert $\phi(e)$ into HashExpressor in the order of $\{h_4,h_2,h_7\}$.

\textbf{2) Query.}
To retrieve the hash function set $\phi(e)$ for a key $e$, HashExpressor maps $e$ to the Cell $c_1$ with the predefined hash function $f$.
If $c_1$ is $empty$, $e$ has not adjusted hash functions and the query procedure fails, $\phi(e)=H_0$.
Otherwise, we store hash function $h_{c_1}$ from $c_1.hashindex$ into $\phi(e)$, and then map $e$ to the next cell $c_2$ with $h_{c_1}$.
The procedure repeats until the size of $\phi(e)$ reaches $k$ and the $endbit$ of the last mapped cell $c_k$ is $1$.
If so, $\phi(e) = \{h_{c_1}, h_{c_2},..., h_{c_k}\}$ or $\phi(e)=H_0$.
For example, as shown in Fig. \ref{fig:hashexpressor}(c), we set $k=3$.
If $e$ is mapped along the solid line, $e$ is first mapped to $\left \langle 0, h_4\right \rangle$ with $f$ and we get $h_4$, then $e$ is mapped to $\left \langle 0, h_2\right \rangle$ with $h_4$ and we get $h_2$.
At last, $e$ is mapped to $\left \langle 1, h_7\right \rangle$ with $h_2$ and we get $h_7$, since the $endbit$ of the last cell is $1$, so $\phi(e) = \{ h_4, h_2, h_7\}$.
However, if $e$ is mapped along the dotted line, $e$ is mapped to an empty cell with $h_9$, so $\phi(e) = H_0$.
Note that HashExpressor may suffer from insertion failure when two different keys are mapped to the same cell but can not share the cell space to store hash functions.
To make HashExpressor more compact, we propose a two-phase joint optimization algorithm to tackle the insertion procedure of HashExpressor together with the optimization procedure of hash function selection simultaneously.

\subsection{Two-phase Joint Optimization}
\begin{figure}[t]
    \vspace{6.5mm}
	\centering
	\setlength{\abovecaptionskip}{5pt}
	\includegraphics[width=0.82\linewidth]{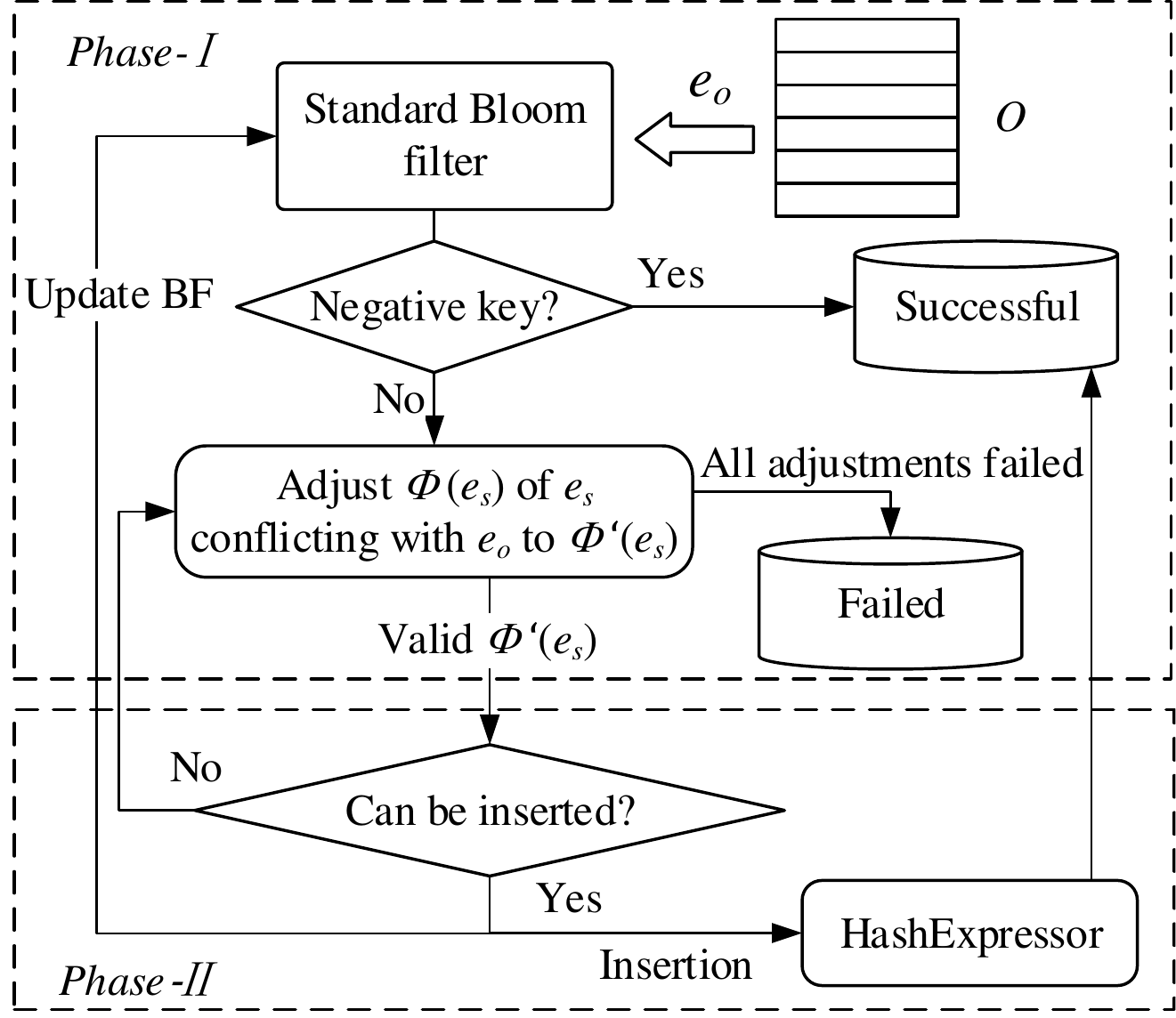}
	\vspace{0.2mm}
	\caption{Procedure of TPJO}
	\label{fig:TPJO}
	\vspace{-10.5mm}
\end{figure}
In this subsection, we introduce the proposed \underbar{T}wo-\underbar{P}hase \underbar{J}oint \underbar{O}ptimization (TPJO) algorithm, including a phase of adjusting hash functions for positive keys ($phase$-$\uppercase\expandafter{\romannumeral1}$) and a phase of inserting the adjusted results into HashExpressor ($phase$-$\uppercase\expandafter{\romannumeral2}$).

We first describe the high-level design of TPJO algorithm.
As shown in Fig. \ref{fig:TPJO}, we initialize standard Bloom filter by inserting all positive keys $e_s$ in $S$ with $H_0$.
In $phase$-$\uppercase\expandafter{\romannumeral1}$, for each key $e_o$ in $O$, we judge whether $e_o$ is tested to be a negative key.
If yes, there is no need to optimize $e_o$.
Otherwise, we adjust $\phi(e_s)$ of $e_s$ to $\phi^{'}(e_s)$, where $e_s$ conflicts with $e_o$.
If then $e_o$ can be tested to be negative, we denote $\phi^{'}(e_s)$ as valid.
In $phase$-$\uppercase\expandafter{\romannumeral2}$, we test whether the valid $\phi^{'}(e_s)$ could be inserted into HashExpressor.
If yes, we insert $\phi^{'}(e_s)$ into HashExpressor and update the Bloom filter.
Otherwise, the insertion in $phase$-$\uppercase\expandafter{\romannumeral2}$ fails, and then we obtain a new $\phi^{'}(e_s)$ in $phase$-$\uppercase\expandafter{\romannumeral1}$.
When $e_o$ is always tested as a positive key whatever $\phi^{'}(e_s)$ is or all valid $\phi^{'}(e_s)$ cannot be inserted into HashExpressor, the optimization of $e_o$ fails.
Besides, the probability of insertion failure for HashExpressor will increase as the number of inserted keys increases.
Therefore, in $phase$-$\uppercase\expandafter{\romannumeral1}$, we first turn to optimize the negative keys with high cost.

We now introduce how to implement TPJO algorithm in detail.
For a key $e_o$ in $O$, according to whether $e_o$ conflicts with keys in $S$, we divide $e_o$ into two types: collision key $e_{ck}$ and optimized key $e_{opk}$.
We first define two runtime auxiliary data structures:
one is the index of bits in Bloom filter that are only mapped by a single positive key in $S$ and only once, we denote it as $V$;
the other is the index of bits that are mapped by optimized key $e_{opk}$, and we denote it as $\Gamma$.

To avoid performance degradation caused by too many adjustment operations, we consider adjusting hash functions of positive keys from $V$.
Let $m$ be the number of bits in Bloom filter. As shown in Fig. \ref{fig:V}, $V$ is composed of $m$ units, and each unit corresponds to one bit in Bloom filter with the same position and is used to store a $2$-tuple: $\left \langle singleflag, keyid\right \rangle$.
Let $V[i]$ be the $i^{th}$ unit in $V$, $V[i].singleflag$ and $V[i].keyid$ represent the $singleflag$ and $keyid$ of $V[i]$, respectively.
$V[i].singleflag$ indicates whether $V[i]$ is mapped by positive keys at most once, $V[i].keyid$ is used to store the identifier (\eg, a pointer in C++) of $e_s$ which is mapped to $V[i]$ first.
We initialize the value of $V[i].singleflag$ to $1$ and $V[i].keyid$ to $NULL$.
\begin{figure}[t]
    \vspace{-0pt}
	\centering
	\setlength{\abovecaptionskip}{5pt}
	\includegraphics[width=0.9\linewidth]{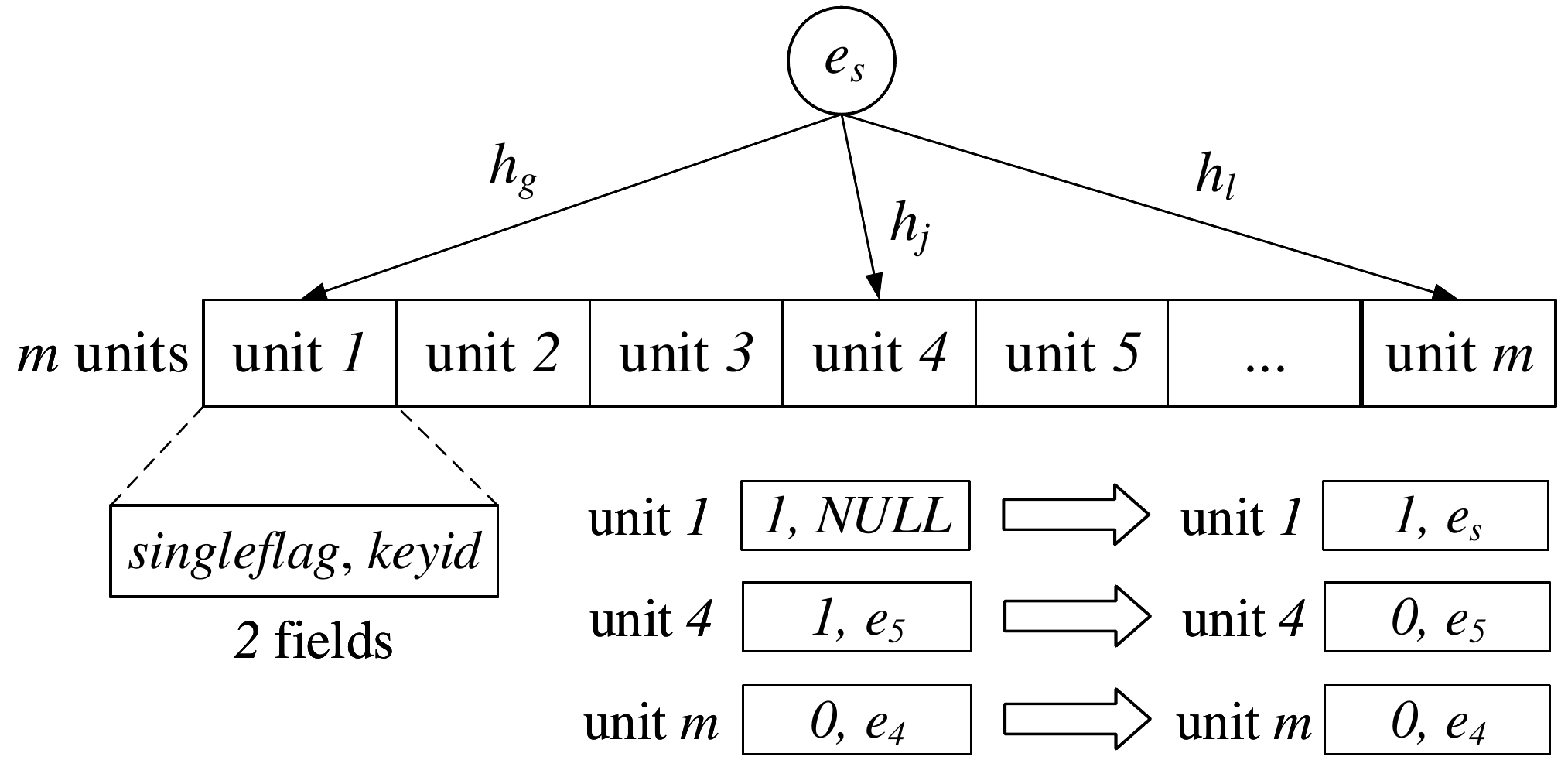}
	\caption{ Data structure of $V$}
	\label{fig:V}
	\vspace{-3mm}
\end{figure}
\begin{figure}[t]
	\centering
	\setlength{\abovecaptionskip}{5pt}
	\includegraphics[width=0.98\linewidth]{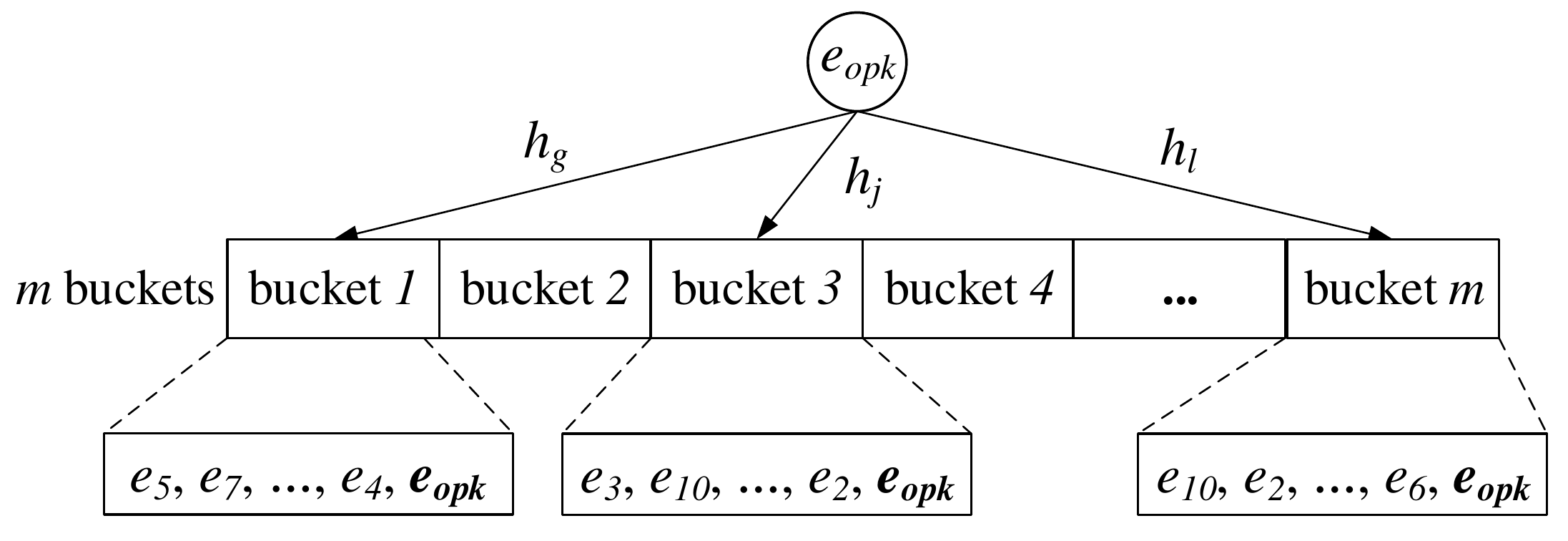}
	\vspace{-1mm}
	\caption{ Data structure of $\Gamma$}
	\label{fig:T}
	\vspace{-14mm}
\end{figure}
To construct $V$, we randomly insert all positive keys in $S$ into $V$.
For a certain key $e_s$ in $S$, $e_s$ needs to be inserted $k = |\phi(e_s)|$ times since it has $k$ hash functions, and when $e_s$ is inserted into a unit $u$, there are three cases:

\emph{Case 1:} If $u.singleflag=1$ and $u.keyid = NULL$, the identifier of $e_s$ is inserted into $u.keyid$, \eg, unit $1$ in Fig. \ref{fig:V}, which changes from $\left \langle 1, NULL\right \rangle$ to $\left \langle 1, e_s\right \rangle$.

\emph{Case 2:} If $u.singleflag=1$ and $u.keyid \neq NULL$, we set $u.singleflag=0$. \eg, unit $4$ in Fig. \ref{fig:V}, which changes from $\left \langle 1, e_5\right \rangle$ to $\left \langle 0, e_5\right \rangle$.

\emph{Case 3:} If $u.singleflag=0$, no changes to $u$, \eg, unit $m$ in Fig. \ref{fig:V}.

\vspace{-1mm}The latter two cases indicate that unit $u$ will be mapped at least twice by the positive keys in $S$.
Then, we use $\Gamma$ to gather all $e_{opk}$s which change to collision keys due to the update of $\phi(e_s)$ in Fig. \ref{fig:TPJO}.
As shown in Fig. \ref{fig:T}, $\Gamma$ is composed of $m$ buckets and $\Gamma[i]$ represents the $i^{th}$ buckets in $\Gamma$.
Each bucket corresponds to the bit in Bloom filter with the same position, and stores identifiers of all $e_{opk}$s mapped to it.
For each bucket $\nu$, we conduct conflict detection for $\nu$ in Algorithm \ref{alg:deconflict}.
If the bit in Bloom filter corresponding to $\nu$ changes from $0$ to $1$, this operation gathers all conflicting $e_{opk}$s in $\nu$ as set $\zeta_\nu$.
%
%\begin{minipage}[t]{0.45\textwidth}
%\vspace{-3mm}
%\vspace{-3mm}
\begin{algorithm2e}[t]
\small
	\BlankLine
	\caption{Conflict Detection}\label{alg:deconflict}
	\vspace{-1mm}
	\KwIn{Bucket $\nu$}
	\KwOut{Conflict optimized key set \emph{\textbf{$\zeta$}}$_\nu$.}
	\vspace{-1mm}
	\BlankLine
	\For{$e_{opk} \in \nu$}{
        $count$ = 0\\
		\For{$h \in \phi(e_{opk})$}{
			\textbf{if} $\Gamma[h(e_{opk})] \neq \nu$ and $V[h(e_{opk})].keyid \neq NULL$ \textbf{then}
				$count$++
            \textbf{endif}
		}
        \textbf{if} $count$ is $k-1$ \textbf{then}
			Add $e_{opk}$ to \emph{\textbf{$\zeta$}}$_\nu$
		\textbf{endif}
	}
	\Return \emph{\textbf{$\zeta$}}$_\nu$.
\end{algorithm2e}
%\end{minipage}

%
Next, we describe how to select hash functions for $\phi(e_s)$ in $phase$-$\uppercase\expandafter{\romannumeral1}$ to specifically optimize $e_{ck}$s.
As shown in Fig. \ref{fig:phase1}, Collision Queue (abbreviated as $CQ$ below) represents the queue composed of $e_{ck}$s to be optimized, which are arranged in descending order of cost.
When optimizing a certain collision key $e_{ck}$, $e_{ck}$ is first mapped to $V$ by $H_0$ to obtain units that meet the following conditions:
\begin{align*}
 singleflag = 1 \wedge keyid \neq NULL.
\end{align*}
Let $\xi_{ck}$ denote the set of these units, for any $u \in \xi_{ck}$, it is only mapped once by a single positive key, and we get $e_s$ by $u.keyid$.
Let $h_u$ be the hash function where $e_s$ is mapped to $u$ by $h_u$, and $H_c$ be the candidate hash functions set, namely $H_c = H-\phi(e_s)$, we conduct an adjustment operation: using one hash function in $H_c$ to replace $h_u$ in $\phi(e_s)$.
\begin{figure}[t]
    \vspace{0mm}
	\centering
	\setlength{\abovecaptionskip}{5pt}
	\includegraphics[width=0.9\linewidth]{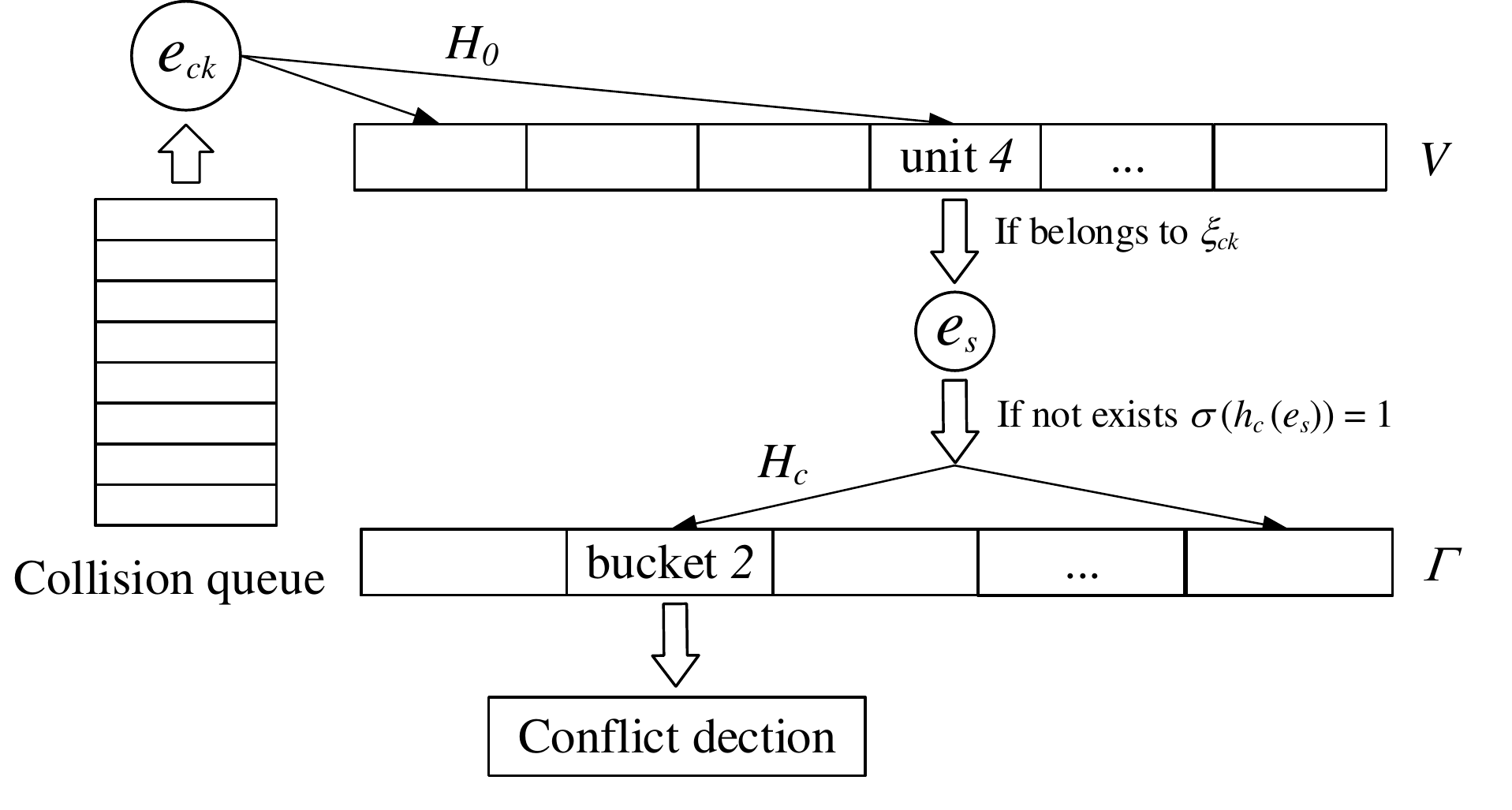}
	\vspace{-2mm}
	\caption{Procedure of $phase$-$\uppercase\expandafter{\romannumeral1}$}
	\label{fig:phase1}
    \vspace{-14mm}
\end{figure}

If there exists a hash function $h_c$ in $H_c$ where $\sigma(h_c(e_s))=1$, $e_{ck}$ can be optimized directly by replacing $h_u$ with $h_c$ without generating new collision keys.
Otherwise, we map $e_s$ to $|H_c|$ buckets of $\Gamma$ and conduct conflict detection for each bucket.
When detecting bucket $\nu$, if $\zeta_\nu \neq \emptyset$, we call $\nu$ \emph{conflict after adjustment}, which means adding $h_c$ to $\phi(e_s)$ will make $e_{opk} \in \zeta_\nu$ become a collision key.
For convenience, we also denote $\Theta(\nu)$ as the overall cost of all conflicting optimized keys in bucket $\nu$.
If there is a bucket that is not \emph{conflict after adjustment}, we can easily use the mapped hash function to replace $h_u$ in $\phi(e_s)$.
Otherwise, we denote $\nu'$ as the bucket with the largest non-negative value of $(\Theta(e_{ck}) - \Theta(\nu'))$.
To minimize the weighted FPR in Equation (\ref{model:equa1}), we choose the hash function mapped to $\nu'$ to replace $h_u$ in $phase$-$\uppercase\expandafter{\romannumeral1}$.
%
%The replacement operation in this case is short for $cost$-$exchange$.
%
In particular, if all buckets are \emph{conflict after adjustment}, and $\Theta(e_{ck}) < \Theta(\nu)$ for any bucket $\nu$, there is no need to optimize $e_{ck}$ as it will bring more cost.

For convenience, we define $e_s\in \xi_{ck}$ if $e_s=u.keyid, u\in \xi_{ck}$.
If we can optimize $e_{ck}$ and insert $\phi(e_s)$ into HashExpressor successfully, we insert $e_{ck}$ into $\Gamma$ and update $V$.
Specifically, for updating $V$, we reset unit $u$ and insert $e_s$ into a new unit by the exchanged hash function.
Besides, if the adjustment generates new collision keys in $phase$-$\uppercase\expandafter{\romannumeral1}$, we insert them into the tail of $CQ$.

\textbf{Example:} As shown in Fig.~\ref{fig:example}, we set $k=3$, $H=\{h_1, h_2, h_3, h_4, h_5, h_6\}$, $H_0=\{h_1, h_2, h_3\}$.
When optimizing a collision key $e_1$, $e_1$ is first mapped to three units in $V$, the $singleflag$ of $\left \langle 1, e_7\right \rangle$ is $1$, which means it is only inserted by $e_7$ once.
Therefore, we consider adjusting the hash functions of $e_7$.
Let $\phi(e_7) = H_0$ and $h_2$ be the hash function of $e_7$ to be mapped to $\left \langle 1, e_7\right \rangle$, we use hash functions in $H_c=\{h_4, h_5, h_6\}$ to replace $h_2$ of $\phi(e_7)$.
We assume that only $\sigma(h_4(e_7))=1$, so one selection of $\phi(e_7)$ is $\{h_1, h_3, h_4\}$.
Then we use $h_5$ and $h_6$ to map $e_7$ to two buckets and conduct conflict
detection respectively.
For the first bucket, we assume there is no confliction for already optimized keys after adjustment, which indicates $\{h_1, h_3, h_5\}$ is also a selection for $\phi(e_7)$.
For the second bucket, we assume $e_2$ is conflicted after adjustment and $\Theta(e_2)>\Theta(e_1)$, so $\{h_1, h_3, h_6\}$ is not a selection.
Therefore, to optimize $e_1$, there are two candidate adjustment selections for $e_7$, and if both of them can not be stored, then $e_1$ fails to be optimized.
Otherwise, among the two choices (\ie, $\{h_1, h_3, h_4\}$ and $\{h_1, h_3, h_5\}$), we store the one with maximized overlap (with hash functions already stored in HashExpressor) into HashExpressor.
%
%As for the specific choice in this example, we
%
\begin{figure}[H]
    \vspace{-1mm}
	\centering
	\setlength{\abovecaptionskip}{5pt}
	\includegraphics[width=0.9\linewidth]{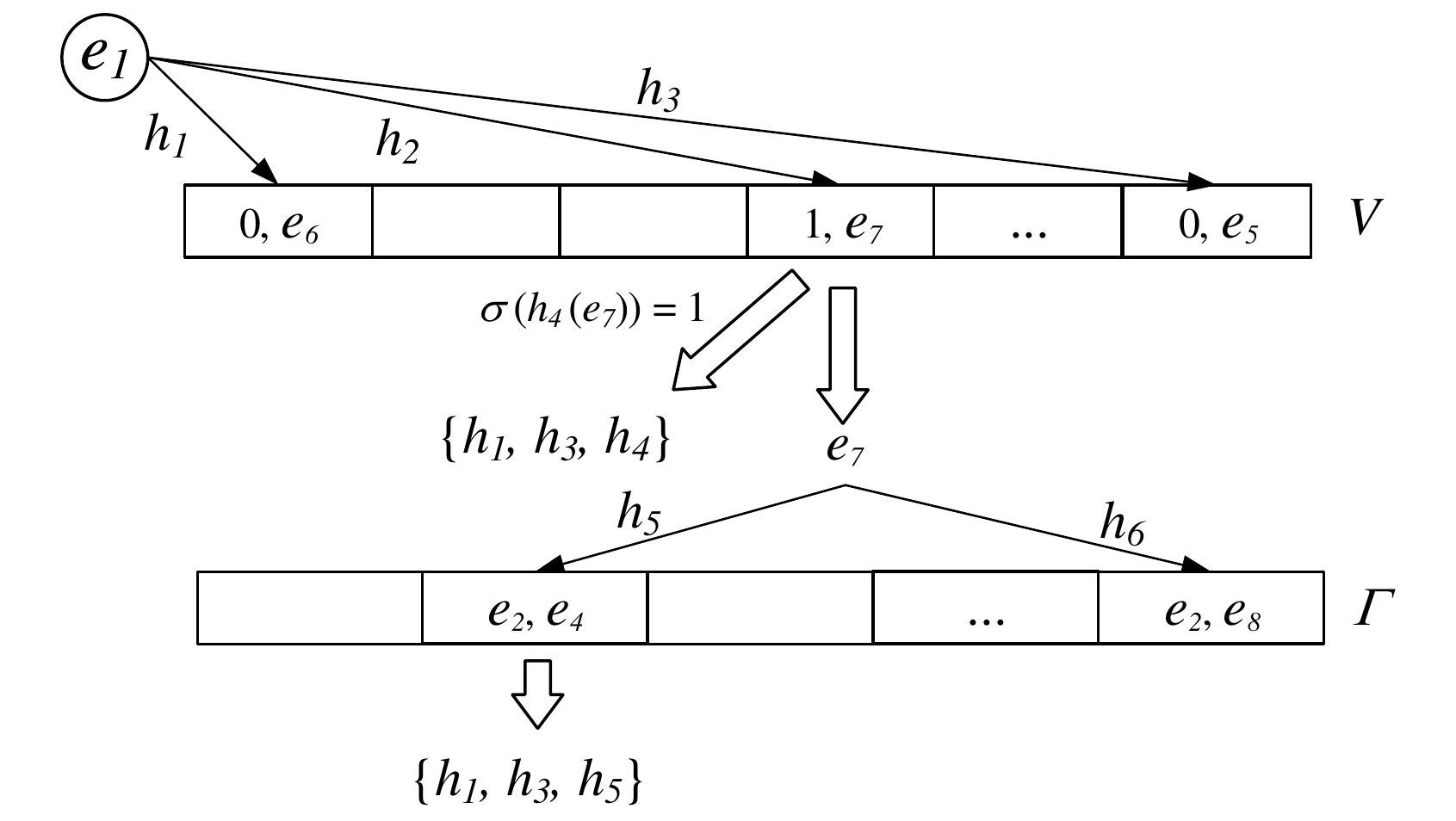}
    \vspace{-1mm}
	\caption{An example for optimizing a collision key}
	\label{fig:example}
    \vspace{-7mm}
\end{figure}
\subsection{Zero-FNR Query}
As mentioned before, HashExpressor is a lightweight hash table, and it has a zero FNR and a small FPR.
Specifically, let $c_k$ represent the last mapped cell in HashExpressor when querying a certain key $e$.
If $e$ has been inserted into HashExpressor, $e$ will definitely get its hash function selection (zero FNR).
Otherwise, $\phi(e) = H_0$.
However, if all cells mapped by $e$ are not empty due to conflicts and $c_k.endbit = 1$ during query, $e$ will be misjudged as an inserted key and the queried $\phi(e) \neq H_0$
which means HashExpressor has a small FPR.
%
%What's worse, when HABF conducts membership testing according to the structure in Fig. \ref{fig:fig3}, if an error occurs when querying hash functions of the positive keys in HashExpressor, it will cause FNR.
%%

%
To make HABF provide the same query pattern (\ie, zero FNR) as Bloom filter, we propose a two-round query mechanism as shown in Fig.~\ref{fig:HABFstructure}.
To be specific, $e$ first uses $H_0$ to check whether it is positive.
If yes, $e$ is identified as a positive key.
If no, we query $\phi(e)$ from HashExpressor and conduct second-round query by using $\phi(e)$ to check again.
If yes, $e$ is also identified to be positive otherwise $e$ is negative.

Next, we analyze how this two-round query pattern achieves zero FNR in HABF.
For a certain key $e_s$ in $S$, there are two possible cases: 1) if $e_s$ has not been inserted into HashExpressor, \ie, $\phi(e_s) = H_0$, $e_s$ will be correctly tested to be positive by the first-round query, and 2) if $e_s$ has been inserted into HashExpressor, namely $\phi(e) \neq H_0$, since HashExpressor has zero FNR, $e_s$ will get $\phi(e)$ correctly, and $e_s$ will be also tested to be positive by the second-round query.
Therefore, HABF achieves Zero-FNR Query.

\subsection{FPR Analysis}
Based on the two-round query pattern, we analyze the FPR of HABF, which is denoted as $F_{habf}$.
Let $F^*_{bf}$ represent the FPR of Bloom filter after optimization and $F_{h}$ represent the FPR of HashExpressor.
For a certain key $e_o$ in $O$, we discuss how $e_o$ will be correctly tested to be negative.
If $e_o$ is tested to be negative by $H_0$ in the first-round query, HashExpressor will query its $\phi(e_o)$ in the second-round query.
If HashExpressor gives the correct result, namely the queried $\phi(e_o) = H_0$, Bloom filter will test $e_o$ by $H_0$ again, then $e_o$ will be finally tested to be negative.
Otherwise, if HashExpressor gives an incorrect result $\phi'(e_o)$, Bloom filter will test $e_o$ with $\phi'(e_o)$.
If $e_o$ can be tested to be negative by $\phi'(e_o)$, $e_o$ will be finally tested to be negative, too.
So $F_{habf}$ can be expressed as
\begin{align} \label{model:equa3}
F_{habf} &= 1-(1-F^*_{bf})(1-F_h+F_h\cdot (1 - F^*_{bf}))\notag\\
   &= F^*_{bf} + F_h\cdot F^*_{bf} - F_h(F^*_{bf})^2.
\end{align}

%
%\vspace{-2mm}
For $F_h$, given a HashExpressor with $\omega$ cells, if $t$ keys have been inserted into HashExpressor, there are at most $t$ cells of which $endbit$ is set to $1$.
For simplification, we assume that the $endbit$s of these cells are evenly distributed, then when querying a key $e$, for the last mapped cell $c_k$, the probability of $c_k.endbit$ being 1 is less than or equal to $\frac{t}{\omega}$.
Therefore, $F_h \leq \frac{t}{\omega}$, and we can derive that $F_{habf} \leq \frac{(\omega+t)}{\omega}\cdot F^*_{bf}$.
In fact, in the actual optimization process of HABF, if we set $t$ to be much smaller than $\omega$, then we have $F_{habf} \approx F^*_{bf}$.
%
%Next, we analyze the efficiency for the deployment of HABF. As for query time, there are at most two queries for Bloom filter and one query for HashExpressor, all of which take a short time for  calculation. As for construction time,
%
\subsection{Discussion: Fast Construction and Query.}%
Considering that there is much hash function computation during the optimization of HABF, heavy computation overhead will inevitably be introduced if a quiet number of hash functions are used.
Inspired by~\cite{kirsch2006less}, we employ double hashing for some scenarios.
To be specific, we reduce hash function calculation by simulating a new hash value from two previously calculated hash values $h_1(x)$ and $h_2(x)$, \eg, simulated hash values $g_i(x) = h_1(x) + ih_2(x)$ for $i=0,...,k-1$.
%  as the generated hash values are correlated with each other and thus not fully random
Note that the double-hashing technique may lead to performance degradation~\cite{dillinger2010adaptive}.
However, targeting at higher (query/construction) throughput, we provide a fast version HABF with double hashing and denote it as \textit{f-HABF}.
Further, for faster construction in some scenarios, \textit{f-HABF} speeds up the procedure at the expense of sacrificing partial hash function selections by disabling $\Gamma$ which contains complex operations for accuracy.
%\red{
%%
%There are many works \cite{guo2009dynamic, deng2006approximately, bhattacharya2020adaptive} studying in this scenario by either sacrificing the feature of Zero-FNR, or expanding the space size of filter structure.
%}
%% 

%% file: Latex/analysis.tex
\section{Theoretical Analysis}
\label{sec:analysis}
In this section, we theoretically analyze the performance of HABF.
We give the theoretical bound for the expected number of collision keys that HABF can optimize in $CQ$.
%
%For a certain collision key $e_{ck}$, we give the mathematical formula of the probability that $e_{ck}$ can be successfully optimized.
%%
%Let $F_{bf}$ represent the FPR of Bloom filter before optimized, since HABF only optimizes the $e_{ck}$s, which means $F^*_{bf} \leq F_{bf}$.
%
Then, we derive the formula of the theoretical bound of $F^*_{bf}$.
\subsection{Analysis for Probability $P_{\xi}$}
In $phase$-$\uppercase\expandafter{\romannumeral1}$, to avoid performance degradation caused by too many adjustment operations, we only adjust the hash functions of positive keys from units in $V$ that are inserted only once.
For a certain Collision Key $e_{ck}$, these positive keys constitute the set $\xi_{ck}$.
For any unit $u$ in $V$, we first analyze the probability $P_{\xi}$ that $u \in \xi_{ck}$.

\begin{theorem}
\label{theorem:1}
If $e_{ck}$ is a collision key, $b$ is the number of bits allocated for each key, when $e_{ck}$ is mapped to a unit $u$ in $V$, for the probability $P_{\xi}$ that $u$ belongs to $\xi_{ck}$, we have
\vspace*{-1mm}
\begin{align}
  E(P_{\xi})> \frac{\frac{k}{b}}{e^{\frac{k}{b}}-1}.
\end{align}
%\vspace*{-4mm}
\end{theorem}
\begin{proof}
For a certain hash function $h$, we assume that $h$ satisfies a distribution $p$ ($p$ may be non-uniform). For any key $e$, the probability $p(u)$ that $u$ is mapped by $e$ using $h$ is determined by its distribution $p$.
Let $P_1(u)$ represent the probability that $u$ is only inserted once while all positive keys are inserted into $V$ with $k$ times.
%
%In fact, the number of adjusted positive keys is usually very small.
%
We assume that the distribution of $P_1(u)$ is approximately constant during optimization.
For convenience, we define $p\in H_0$ if $p$ is the distribution of the corresponding hash function in $H_0$.
Moreover, we assume that the hash functions are independent of each other.
Then, $P_1(u)$ can be expressed as
\begin{align}
 P_1(u)= &|S|\cdot( \sum \limits_{p\in H_0}p(u)\prod \limits_{p'\in H_0,p'\neq p}(1-p'(u))) \notag\\ & \cdot (\prod \limits_{p\in H_0}(1-p(u)))^{|S|-1} \notag\\
 %= &|S|\cdot \sum \limits_{p\in H_0}\frac{p(i)}{1-p(i)}\prod \limits_{p'\in H_0}(1-p'(u))\notag\\ & \cdot (\prod \limits_{p\in H_0}(1-p(u)))^{|S|-1}\notag\\
 > &|S|\cdot \sum \limits_{p\in H_0}p(u)(\prod \limits_{p\in H_0}(1-p(u)))^{|S|}.
\end{align}
%
%\textbf{Discussion:} In actual implementation, the design of the hash function is based on characters of keys, thus there exists a weak relationship between hash functions. While as the number of keys increases, the hash functions are closer to being independent of each other. To simplify the analysis, here we consider the ideal situation where the hash functions are independent of each other.

Let $P_0(u)$ represent the probability that $u$ is empty, then $P_0(u) = 1-(\prod \limits_{p\in H_0}(1-p(u)))^{|S|}$. 
As the units mapped by $e_{ck}$ are inserted at least once, $P_{\xi}(u)$ can be expressed as a conditional probability:
\begin{align}\label{equa:2}
 P_{\xi}(u) = \frac{P_1(u)}{1-P_0(u)} & > \frac{|S|\cdot \sum \limits_{p\in H_0}p(u)}{\frac{1}{(\prod \limits_{p\in H_0}(1-p(u)))^{|S|}}-1}.
\end{align}

\begin{lemma}\label{lemma:1}
$\forall p \in H_0$, $0\leq p(u)\leq 1$, we have
\begin{align}
\label{equal:lemma:1}
\prod \limits_{p\in H_0}(1-p(u)) \geq 1- \sum \limits_{p\in H_0}p(u).
\end{align}
\end{lemma}
%\begin{proof}
%Let $p_i$ be the distribution of the hash function $h_i$, then Equation (\ref{equal:lemma:1}) can be expressed as
%\begin{align}
%\label{equa:appendixl1-2}
%\prod \limits_{i = 0}^{k}(1-p_i(u)) \geq 1- \sum \limits_{i = 0}^{k}p_i(u).
%\end{align}
%%
%
%%%
%We denote Equation (\ref{equa:appendixl1-2}) as $\Psi$.
%%
%Next, we use mathematical induction to prove $\Psi$. Obviously, it holds when $k=0$; then, we assume that $\Psi$ holds when $k=\alpha-1$, thus we have $\prod \limits_{i = 0}^{\alpha-1}(1-p_i(u)) \geq 1- \sum \limits_{i = 0}^{\alpha-1}p_i(u)$ and thereby we get
%\begin{align}
%\label{equa:appendixl1-3}
%\prod \limits_{i = 0}^{\alpha}(1-p_i(u)) & = (1-p_\alpha(u))\prod \limits_{i = 0}^{\alpha-1}(1-p_i(u))\notag\\
%&=\prod \limits_{i = 0}^{\alpha-1}(1-p_i(u)) - p_\alpha(u) \prod \limits_{i = 0}^{\alpha-1}(1-p_i(u))\notag\\
%%& \geq 1- \sum \limits_{i = 0}^{\alpha-1}p_i(u) - p_\alpha(u) \prod \limits_{i = 0}^{\alpha-1}(1-p_i(u))\notag\\
%& \geq 1- \sum \limits_{i = 0}^{\alpha}p_i(u).
%\end{align}
%\vspace{-1.5mm}
%Therefore, $\Psi$ holds when $k=\alpha$, this completes the proof.
%\end{proof}

\begin{lemma}\label{lemma:2}
$\forall x \in [0,1]$, $f(x) = \frac{|S|\cdot x}{\frac{1}{(1-x)^{|S|}}-1}$
is convex.
\end{lemma}

Due to space limitations, the proofs of Lemma \ref{lemma:1} and Lemma \ref{lemma:2} are detailed in the appendix.
Let $x=\sum \limits_{p\in H_0}p(i)$, as per Lemma \ref{lemma:1}, $P_{\xi}(u) > \frac{|S|\cdot x}{\frac{1}{(1-x)^{|S|}}-1} = f(x)$.
%
%\begin{align}
%P_{\xi}(u) > \frac{|S|\cdot x}{\frac{1}{(1-x)^{|S|}}-1} = f(x).
%\end{align}
As per Lemma \ref{lemma:2}, $f(x)$ is convex, by Jensen inequality \cite{rudin2006real}, we get
\begin{align}
E(P_{\xi}) = E(P_{\xi}(u)) > E(f(x)) \geq f(E(x)).
\end{align}

For any hash function distribution $p$, $E(p(u)) = \frac{1}{m}$, and $E(x) = E(\sum \limits_{p\in H_0}p(i)) = \sum \limits_{p\in H_0}E(p(i)) =  \frac{k}{m}$, so we have
\begin{align}
E(P_{\xi}) %& >  \frac{|S|\cdot E(x)}{\frac{1}{(1-E(x))^{|S|}}-1} \notag\\
%&=  \frac{|S|\cdot \frac{k}{m}}{\frac{1}{(1-\frac{k}{m})^{|S|}}-1} \notag\\
>\frac{|S|\cdot \frac{k}{m}}{(1-\frac{k}{m})^{(-\frac{m}{k})\cdot \frac{k}{m} \cdot |S|}-1} \approx  \frac{\frac{k}{b}}{e^{\frac{k}{b}}-1}.
\end{align}
This completes the proof.
\end{proof}

\subsection{Analysis for $F^*_{bf}$}
Let $F_{bf}$ represent the FPR of Bloom filter before optimization, and since HABF only optimizes the $e_{ck}$s, which means $F^{*}_{bf} \leq F_{bf}$.
Let $t$ be the number of collision keys optimized by HABF. Thus for $F^*_{bf}$, we can derive that
\vspace*{-1mm}
\begin{align}
\label{equl:fbf10}
F^*_{bf} = F_{bf} - \frac{t}{|O|}.
\end{align}
\vspace*{-4mm}

We first analyze $\xi_{ck}$ before $F^*_{bf}$, for $\forall e_{ck} \in CQ$, $e_{ck}$ is first mapped to $k$ units in $V$, as per Theorem \ref{theorem:1}, we have $E(|\xi_{ck}|) = k\cdot E(P_{\xi})$.
%
%It is easy to prove that function $g(x)=\frac{x}{e^x-1}$ is monotonically decreasing, when $k\leq b$ and $E(P_{\xi})>g(1)$.
%
When $k\geq 2$, $E(|\xi_{ck}|) > 2g(1) >1.164$. %, and $E(|\xi_{ck}|)$ will increase as $k$ increases.
%among all units mapped by $e_{ck}$, the expectation of probability that at least one unit belongs to $\xi_{ck}$ will be greater than $1-(1-g(1))^6 = 0.927$.
%
We assume that at least one unit belongs to $\xi_{ck}$, namely $|\xi_{ck}| \geq 1$ ($k \geq 2$) and we consider the worst case of $|\xi_{ck}|=1$.

Let $u_{ck}$ be the single unit in $\xi_{ck}$ and $e_{sk}$ be the key in $S$ corresponding to $u_{ck}.keyid$.
We denote $P_c$ as the probability that $e_{sk}$ can adjust its hash function in $phase$-$\uppercase\expandafter{\romannumeral1}$ and $P_s$ as the probability that $\phi(e_{sk})$ can be inserted into HashExpressor, $P_c$ and $P_s$ are independent of each other.
For the probability $P_{ck}$ that $e_{ck}$ can be optimized, we have
\begin{align}\label{equa:7}
P_{ck} = P_c \cdot P_s.
\end{align}

For each $e_{ck}$ in $CQ$, $phase$-$\uppercase\expandafter{\romannumeral1}$ provides multiple adjustment schemes ($e_s$ and $\phi(e_s)$) to be inserted into cells in HashExpressor until one of them can be inserted.
We assume that the distribution of the inserted cells in HashExpressor will tend to be approximately uniform. %\ie, each cell in HashExpressor is inserted with equal probability.
If $t$ collision keys have been optimized, we have % which means HashExpressor has been inserted $t$ times at most, we have
\begin{align}\label{equa:8}
P_s(t) > \prod \limits_{i=0}^{k-1}(1-\frac{kt+i}{\omega}) > (1-\frac{kt+k}{\omega})^k.
\end{align}
Let $P'_c$ be the probability that $\phi(e_{sk})$ can be adjusted to a valid $\phi^{'}(e_{sk})$ when all keys in $O$ are inserted into $\Gamma$, not just the optimized keys as mentioned before. It is easy to see that $P_c \geq P'_c$, and $P'_c$ is not related to $t$.
Due to space limitations, the analysis of $P'_c$ is detailed in the appendix.
\begin{theorem}
\label{theorem:2}
If $T$ is the size of $CQ$ and $t$ is the number of Collsion Keys optimized by HABF, we have
\begin{align}
E(t) > \frac{T\cdot P'_c(\omega-k^2)}{\omega+T\cdot P'_c\cdot k^2}.
\end{align}
\end{theorem}
\begin{proof}
We denote HABF$'$ as the HABF that changes operations as follows: no matter whether $e_{ck}$ is optimized successfully or not, we insert a virtual positive key with $k$ randomly selected hash functions into HashExpressor.
Let $E'(t) $ be the expected number of collision keys that can be optimized by HABF$'$. It can be seen intuitively that $E(t) \geq E'(t)$.

Next, we analyze $E'(t)$.
Let $P^{(i)}$ be the probability that the $i^{th}$ collision key in $CQ$ is optimized by HABF$'$.
As per Equation (\ref{equa:7}), we have
\begin{align}\label{equa:11}
P^{(i+1)} = P_{ck}(i) \geq P'_c\cdot P_s(i) >  P'_c(1-\frac{k(i+1)}{\omega})^k .
\end{align}
It is easy to prove that function $g'(i) = (1-\frac{k(i+1)}{\omega})^k$ is a convex function, and $P'_c$ is not related to $i$ as mentioned before.
By the Jensen inequality, we have
\begin{align}
E(P^{(i+1)}) > & P'_c\cdot E(g'(i)) >  P'_c\cdot g'(E(i)).
\end{align}
For HABF$'$, the number of inserted keys in HashExpressor is equal to the number of optimized collision keys, $E(i) = E'(t)$, then we have
\begin{align}
E(P^{(i+1)}) > P'_c\cdot g'(E'(t)).
\end{align}

\begin{lemma}\label{lemma:3}
For a random variable $X_i$, $0\leq i\leq n$, the value of $X_i$ is 0 or 1, the probability expectation of $X_i = 1$ is $E(p_i)$, $\forall i,j \in \mathbb{N}, 0\leq i,j\leq n, i\neq j$, $X_i$ and $X_j$ are independent of each other, we have
\begin{align}
E(\sum \limits_{i=0}^nX_i) = \sum \limits_{i=0}^nE(p_i).
\end{align}
\end{lemma}

It is easy to prove Lemma \ref{lemma:3} by mathematical induction.
As per Equation (\ref{equa:8}) and Equation (\ref{equa:11}), $P^{(i+1)}$ is only determined by $i$, so $\forall 0\leq \alpha,\beta \leq n, \alpha\neq \beta$, $P^{(\alpha)}$ and $P^{(\beta)}$ are independent of each other. By Lemma \ref{lemma:3}, we get
\begin{align}
E'(t) = & \sum \limits_{i=0}^TE(P^{(i)}) > T\cdot P'_c \cdot g'(E'(t)).
\end{align}
As per Lemma \ref{lemma:1}, $g'(E'(t)) = (1-\frac{k(E'(t)+1)}{\omega})^k \geq 1-\frac{k^2(E'(t)+1)}{\omega}$, we have $E'(t) > T\cdot P'_c(1-\frac{k^2(E'(t)+1)}{\omega})$, then
\begin{align}
E(t) \geq E'(t) > \frac{T\cdot P'_c(\omega-k^2)}{\omega+T\cdot P'_c\cdot k^2}.
\end{align}
This completes the proof.
\end{proof}

Based on Theorem \ref{theorem:2} and Equation (\ref{equl:fbf10}), we get
\begin{align}
\label{eq:bound}
E(F^*_{bf}) &=  E(F_{bf}) - \frac{E(t)}{|O|} \notag\\
&<  E(F_{bf}) - \frac{T\cdot P'_c(\omega-k^2)}{|O|(\omega+T\cdot P'_c\cdot k^2)}
\end{align}
\subsection{Experimental Verification}
To validate the upper bound of the expected false positive rate of HABF in Equation~(\ref{eq:bound}), we conduct experiments to verify the theoretical bound of $F^*_{bf}$.
As shown in Fig.~\ref{fig:rvatb(a)}, we set bits-per-key $b = 10$ and vary the number $k$ of hash functions from $2$ to $10$.
In Fig. \ref{fig:rvatb(b)}, we set $k = 4$ and vary $b$ from $4$ to $13$. The results show that the theoretical upper bound perfectly holds as it is always larger than the real value.

\begin{figure}
\vspace{0mm}
\centering
\subfigure[Varying number of hash functions]{
\begin{minipage}[t]{0.49\linewidth}
\centering
\label{fig:rvatb(a)}
\includegraphics[width=1\linewidth]{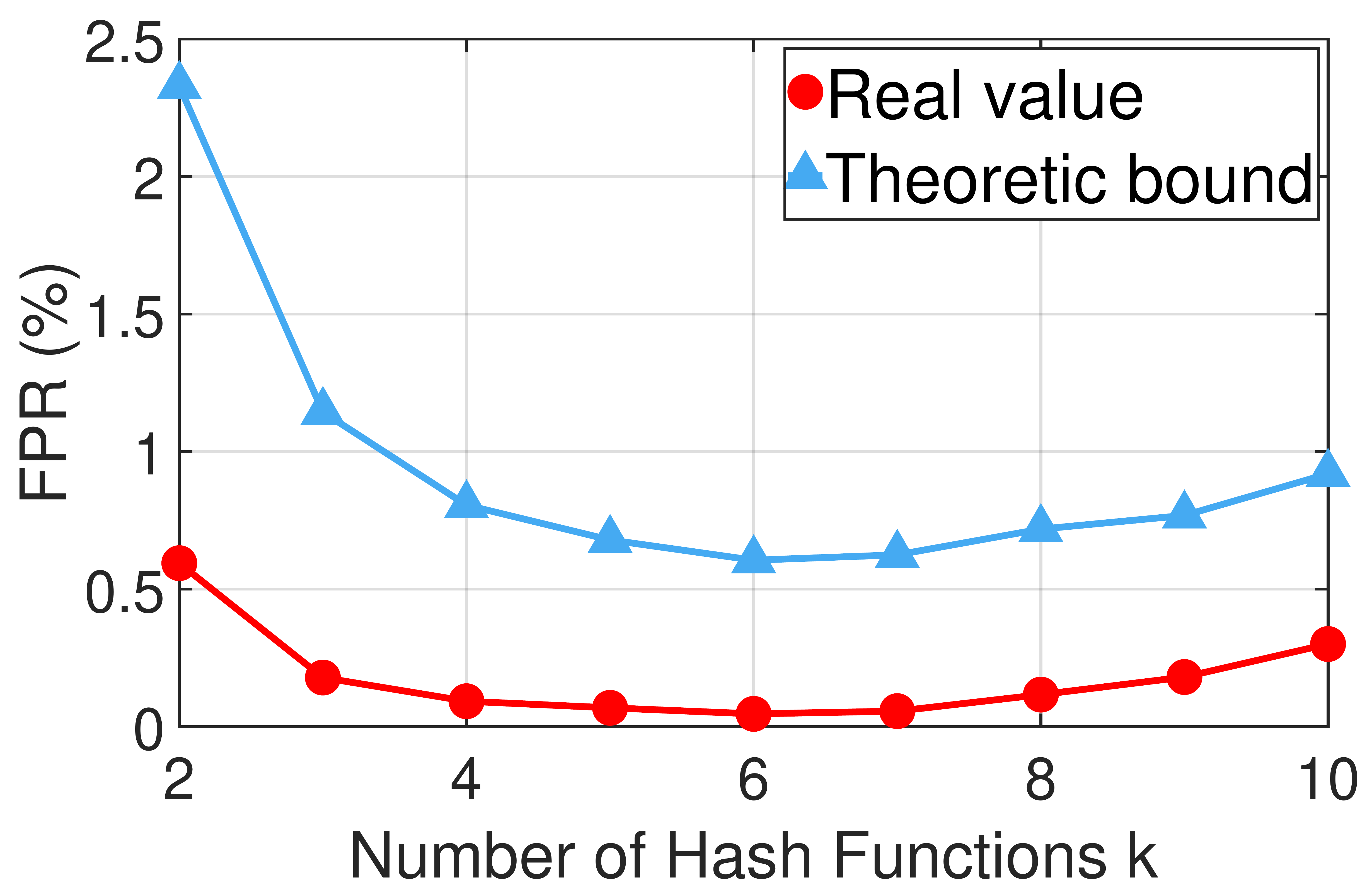}
%\caption{fig1}
\end{minipage}%
}%
\subfigure[Varying bits-per-key]{
\begin{minipage}[t]{0.49\linewidth}
\centering
\label{fig:rvatb(b)}
\includegraphics[width=1\linewidth]{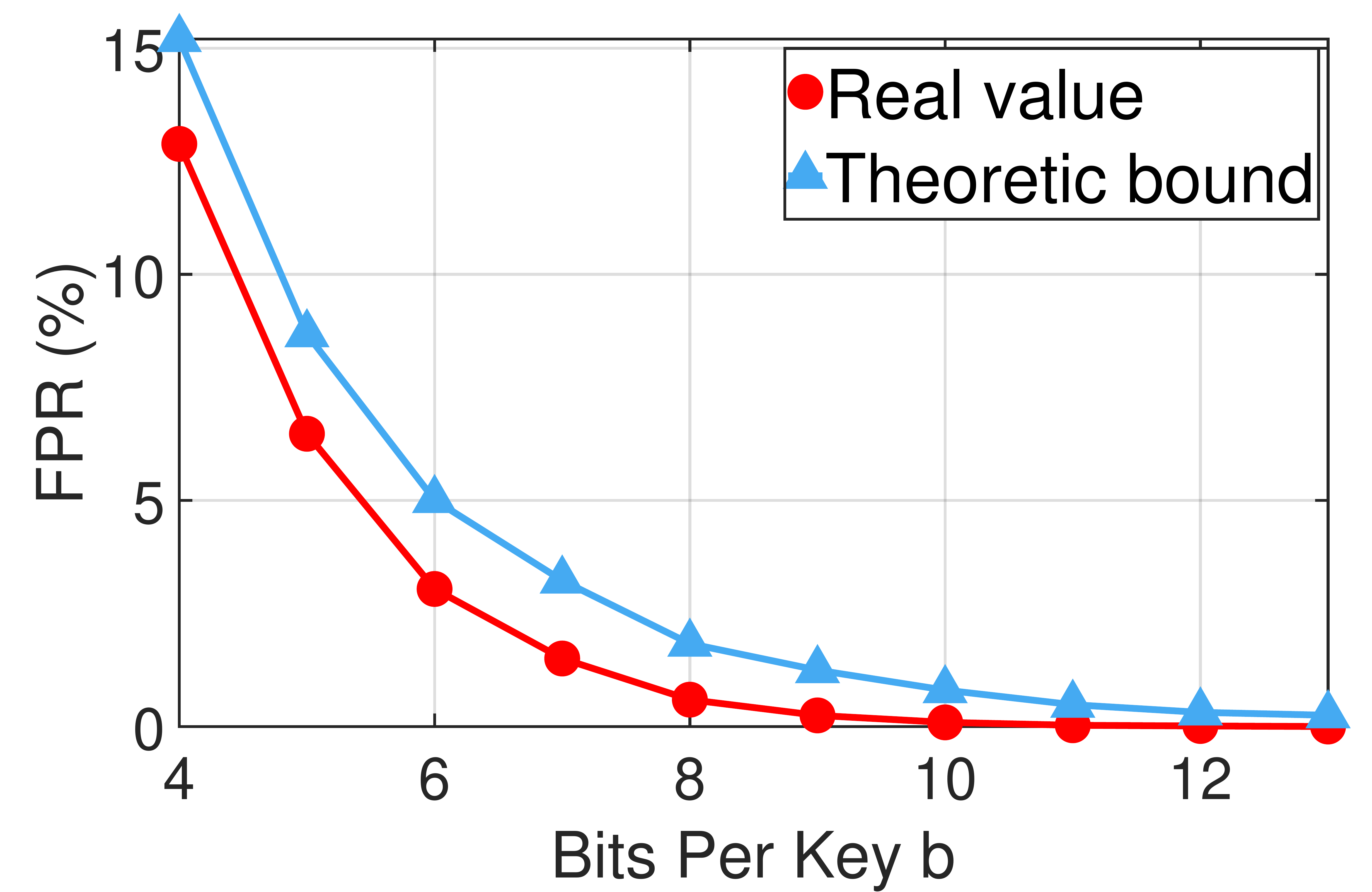}
%\caption{fig2}
\end{minipage}%
}%
\centering
\label{fig:rvatb}
\vspace{-4mm}
\caption{Real value and theoretic bound}
\vspace{-15mm}
\end{figure}

%% file: Latex/evaluation.tex
\section{Experimental Results}
\label{sec:eval}
In this section, we conduct experiments to validate the performance of HABF.
\subsection{Experimental Setup}
The comparison algorithms can be divided into two types:

\textbf{1) Non-learned filter.} We choose a standard Bloom filter (BF) and Xor filter (Xor) \cite{graf2020xor} as baselines.
Given bits-per-key $b$, we set the number of hash functions $k=ln2\cdot b$ to minimize FPR for BF, and set the number of bits of the fingerprint to $\lfloor\frac{b}{1.23+\frac{32}{|S|}}\rfloor$ for Xor.
The optimized implementation comes from \cite{FastFilter}.
%
%Since the performance of Bloom filter will be affected by different implementations of hash function.
%%
%We implement three versions of Bloom filter: BF by using $k$ different kinds of hash functions, BF (City$64$) by using CityHash ($64$bit version) and BF (XXH$128$) by using xxHash ($128$bit version).
%
Besides, under the skewed cost distribution, we also compare HABF with Weighted Bloom filter (WBF).
Considering WBF relies on cost information during the query, thus we cache some keys with high costs in memory for WBF.
%
%To make WBF use the same space as other algorithms, we put $0.1\%$ of the data and weight information in the memory, which together with Bloom filter constitute the space occupied by WBF.

%%
\textbf{2) Learned filter.} Learned filter refers to the set of the state-of-the-art works based on learned index\cite{kraska2018the}, including Learned Bloom filter (LBF)\cite{kraska2018the}, Sandwiched Learned Bloom filter (SLBF)\cite{mitzenmacher2018a}, and Adaptive Learned Bloom filter (AdaBF)\cite{2019adaptive}, which incorporate machine learning (ML) models as the underlying data structures.
For the deep-learning model, we implemented a 16-dimensional character-level RNN (GRU  \cite{cho2014learning}, in particular) and a six-layer fully connected neural network~\cite{hopfield1982neural}, both of which have a $32$-dimensional embedding layer.
Considering that the current learning models generally use GPU to train the model, we also compare the learning model algorithms in the GPU environment, which we denote as LBF (GPU), SLBF (GPU), and AdaBF (GPU).

\textbf{Implementation: } We implement our algorithm and non-learned filter algorithms in C++ and compiled using g++ with -O3 optimization, and learned filter algorithms are implemented in Keras \cite{keras}, which is a deep learning platform.
We summarize all used hash functions and their implementations in Table \ref{Tab:hashfunc}.
If not specified, we set the default hash function used by f-HABF and other algorithms to XXH$128$.
All the programs run on a server with Intel(R) Xeon(R) Gold $6248$ CPU with $10$ cores running at $2.5$GHZ, $106$GB memory, and two Tesla V100 SXM2 GPUs with $32$GB memory.
The source codes of all algorithms are available in \cite{ourSourceCode}.
\begin{table}[t]
\centering
\caption{Hash function set}
\setlength{\tabcolsep}{2mm}
\label{Tab:hashfunc}
\small
\begin{tabular}{|l|c|}
%\toprule
\hline
Hash function & Implementation \\
%\midrule
\hline
{xxHash} & \multirow{1}*{\cite{xxHash}}\\
\hline
CityHash & \multirow{1}*{\cite{CityHash}}\\
\hline
MurmurHash  & \multirow{1}*{\cite{MurmurHash}}\\
\hline
SuperFast, crc32, FNV  & \multirow{1}*{\cite{SMhasher}} \\
\hline
BOB, OAAT  & \multirow{1}*{\cite{Jenkins}} \\
\hline
DEK, Hsieh, PYHash, BRP, TWMX, & \multirow{3}*{\cite{timmurphy.org}}\\
APHash, NDJB, DJB, BKDR, PJW, & \\
JSHash, RSHash, SDBM, ELF &\\
\hline
%\bottomrule
\end{tabular}
\label{tbl:table-example}
\vspace{-3.5mm}
\end{table}
%\begin{table}
%\centering
%\caption{\red{Hash Function Collection}}
%\setlength{\tabcolsep}{4mm}
%\label{Tab:hashfunc}
%\small
%\begin{tabular}{cc|cc}
%\toprule
%Hash Func. & Impl. & Hash Func. & Impl. \\
%\midrule
%xxHash & \cite{xxHash} &  BRPHash & \cite{timmurphy.org} \\
%CityHash & \cite{CityHash} & TWMXHash & \cite{timmurphy.org} \\
%MurmurHash & \cite{MurmurHash} & APHash & \cite{timmurphy.org} \\
%SuperFastHash & \cite{SMhasher} & DJBHash & \cite{timmurphy.org} \\
%crc32Hash & \cite{SMhasher}  & NDJBHash & \cite{timmurphy.org} \\
%FNVHash & \cite{SMhasher} & BKDRHash & \cite{timmurphy.org} \\
%BOB & \cite{Jenkins} & PJWHash & \cite{timmurphy.org} \\
%OAAT & \cite{Jenkins} & JSHash & \cite{timmurphy.org} \\
%DEKHash & \cite{timmurphy.org} & RSHash & \cite{timmurphy.org} \\
%HsiehHash & \cite{timmurphy.org} & SDBMHash & \cite{timmurphy.org} \\
%PYHash & \cite{timmurphy.org} & ELFHash & \cite{timmurphy.org} \\
%\bottomrule
%\end{tabular}
%\label{tbl:table-example}
%\end{table}
%\setlength{\textfloatsep}{0mm}
\subsection{Metrics}
%9
We use the following metrics: (1) weighted FPR; (2) construction time; (3) query latency; and (4) construction memory consumption.
Weighted FPR is defined in Equation (\ref{model:equa1}), \ie, suppose the false positive key set from $O$ is $O'$, then
\begin{align}
\label{euqal:wFPR}
  Weighted \; FPR = \frac{\sum_{e'\in O'} \Theta(e')}{\sum_{e\in O} \Theta(e)},
\end{align}
where $\Theta(e)$ is the cost of $e$.
In particular, if $\Theta$ is a uniform distribution function, weighted FPR is equivalent to traditional FPR.
%
%To show the performance more intuitively, we regard the overall cost of Bloom filter as a baseline, and compare the ratio of the overall cost of all comparison algorithms to it, which we define as the reduction rate.
%
Moreover, the construction time refers to the time to build filters, the query latency refers to the time to conduct membership testing per key, and the construction memory consumption refers to the memory footprint during construction.
To achieve a head-to-head comparison, we set the same bits-per-key for every filter and thus all filters use the same space.
\subsection{Datasets}
\label{eval:datasets}
We use the following two datasets in the experiments:

\textbf{1) Shalla's Blacklists.} Shalla's Blacklists \cite{shalla} is a URL dataset with evident characteristics and available in \cite{DeepBloom}.
The dataset consists of $2.927$ million keys, including $1$,$491$,$178$ positive keys and $1$,$435$,$527$ negative keys.
For simplicity, we call this dataset Shalla for short if no confusion arises.

\textbf{2) YCSB.} YCSB is a benchmark \cite{cooper2010benchmarking} for key-value databases, and we modified YCSB's uniform generator to generate $24$,$074$,$812$ keys, including $12$,$500$,$611$ positive keys and $11$,$574$,$201$ negative keys.
The key schema consists of a $4$-byte prefix and a $64$-bit integer without evident characteristics.

For cost distribution, since all keys in both datasets initially have no cost, we generate Zipf \cite{powers1998applications} distributions with various skewness factors (from $0$ to $3.0$).
In particular, if the skewness factor is $0$, the cost distribution is uniform.
Moreover, for each skewness factor, we randomly shuffle the generated Zipf distribution $10$ times and apply it to each dataset, and then calculate the average weighted FPR.
\subsection{Parameter Performance Evaluation}
We first evaluate the overall performance of HABF.
Let $\Delta_1$ and $\Delta_2$ be the space size of HashExpressor and Bloom filter, and we define the space allocation ratio as $\Delta = \frac{\Delta_1}{\Delta_2}$.
Given the total space size, the performance of HABF is determined by the following three parameters:
(1) space allocation ratio $\Delta$;
(2) number of hash functions $k$;
and (3) cell size of HashExpressor.
Here, we first use Shalla with uniform cost distribution to show how the three parameters affect the performance of HABF.

1) \emph{Effect of $\Delta$.}
We set the space size $\Delta_1 + \Delta_2 = 2$MB and vary $\Delta$ from $0$ to $1$.
The results in Fig. \ref{fig:fig:10(a)} show that when $\Delta$ is low, the failure probability of $\phi '(e)$ to be inserted into HashExpressor increases and the weighted FPR is high.
In particular, if $\Delta=0$, HABF is equivalent to the standard Bloom filter.
When $\Delta$ is high, there will be lots of Collision Keys in Bloom filter, the probability of HABF optimization failure increases and the weighted FPR is high.
The optimal value $\Delta$ is $0.25$, which means that the space allocation ratio between HashExpressor and Bloom filter is $1$:$4$ in this case and will be used as the default parameter throughout the experiments.

2) \emph{Effect of $k$.}
We set the space size $\Delta_1 + \Delta_2 = 2$MB and vary $k$ from $2$ to $8$.
As shown in Fig. \ref{fig:fig:10(a)}, HABF achieves the best performance when $k=3, 4, 5$.
If $k < 3$, the weighted FPR increases since the number of hash functions applied to check a given key decreases.
As $k$ becomes large and $k > 5$, the adjusted hash function sets that could be inserted into HashExpressor will decrease a lot.
$k=3$ is a modest choice, and we set it by default in the following experiments.

3) \emph{Effect of cell size.} The size of a cell is determined by the number of bits in $hashindex$.
If cell size equals $\alpha$, each cell can represent at most $2^{\alpha-1}-1$ hash functions, which is equal to the number $|H|$ of global hash functions.
In our work, we provide $22$ kinds of hash functions and the maximum size of a cell is $5$.
In this setting of experiments, we vary the space size from $1.25$MB to $3.25$MB and compare the performance when cell size equals $3$, $4$, and $5$.
Fig. \ref{fig:fig:10(b)} shows that the weighted FPR is minimized when the cell size equals $4$.
We use this setting by default in the following experiments.
\begin{figure}[t]
\centering
\vspace{0mm}
\subfigure[Weighted FPR vs. $\Delta$ and $k$]{
\begin{minipage}[t]{0.49\linewidth}
\centering
\label{fig:fig:10(a)}
\includegraphics[width=1\linewidth]{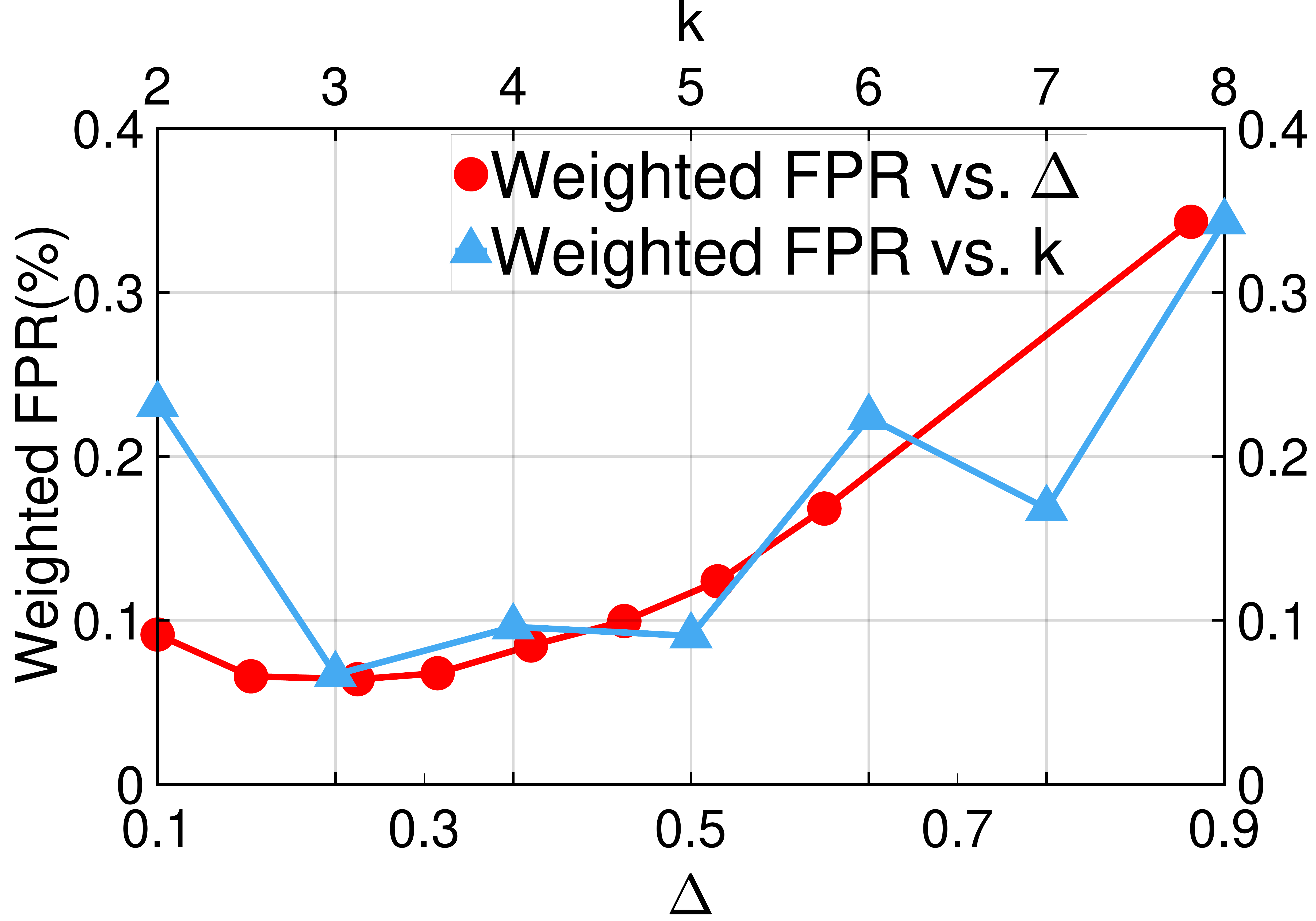}
%\caption{fig1}
\vspace{-2mm}
\end{minipage}%
}%
\subfigure[Weighted FPR vs. cell size]{
\begin{minipage}[t]{0.49\linewidth}
\centering
\label{fig:fig:10(b)}
\includegraphics[width=1\linewidth]{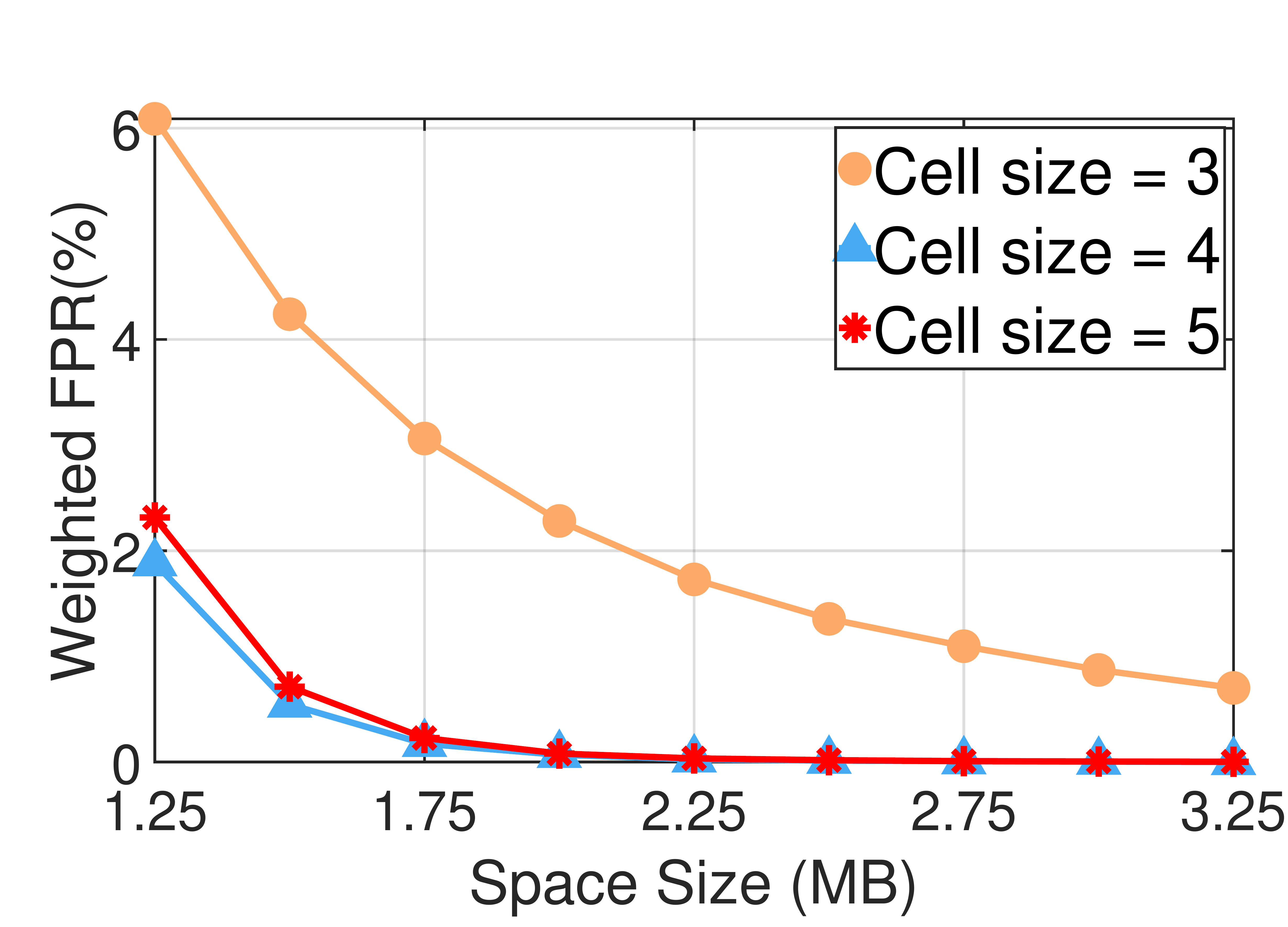}
%\caption{fig2}
\vspace{-2.85mm}
\end{minipage}%
}%
\centering
\label{fig:fig10}
\vspace{-4mm}
\caption{Parameter performance evaluation}
\vspace{-23mm}
\end{figure}
\begin{figure*}[t]
%\vspace{-5mm}
\centering
\subfigure[vs. Non-learned filter (Shalla)]{
\begin{minipage}[t]{0.24\linewidth}
\centering
\label{fig:fig:11(a)}
\includegraphics[width=1\linewidth]{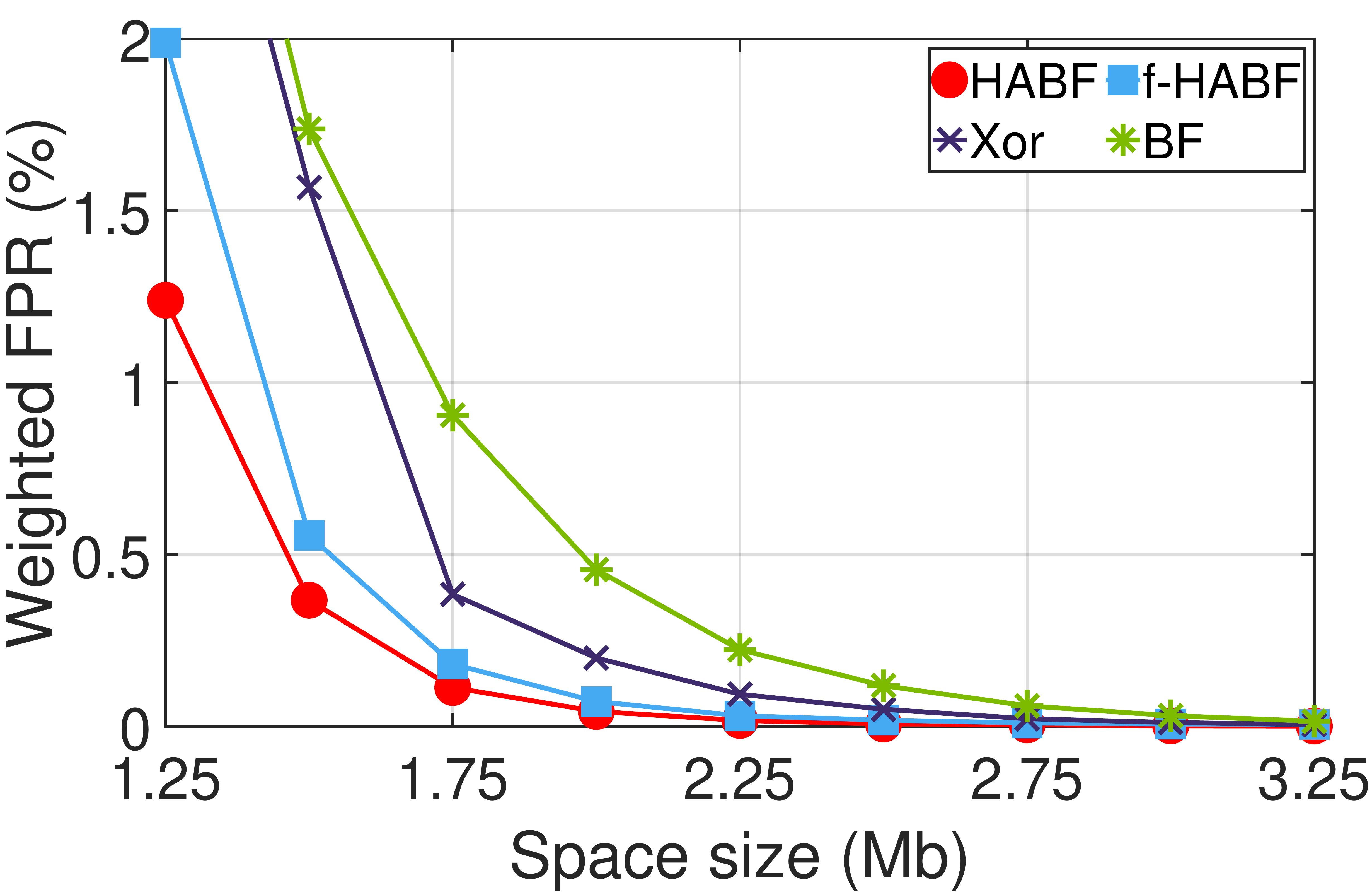}
%\caption{fig1}
\end{minipage}%
}%
\subfigure[vs. Learned filter (Shalla)]{
\begin{minipage}[t]{0.24\linewidth}
\centering
\label{fig:fig:11(b)}
\includegraphics[width=1\linewidth]{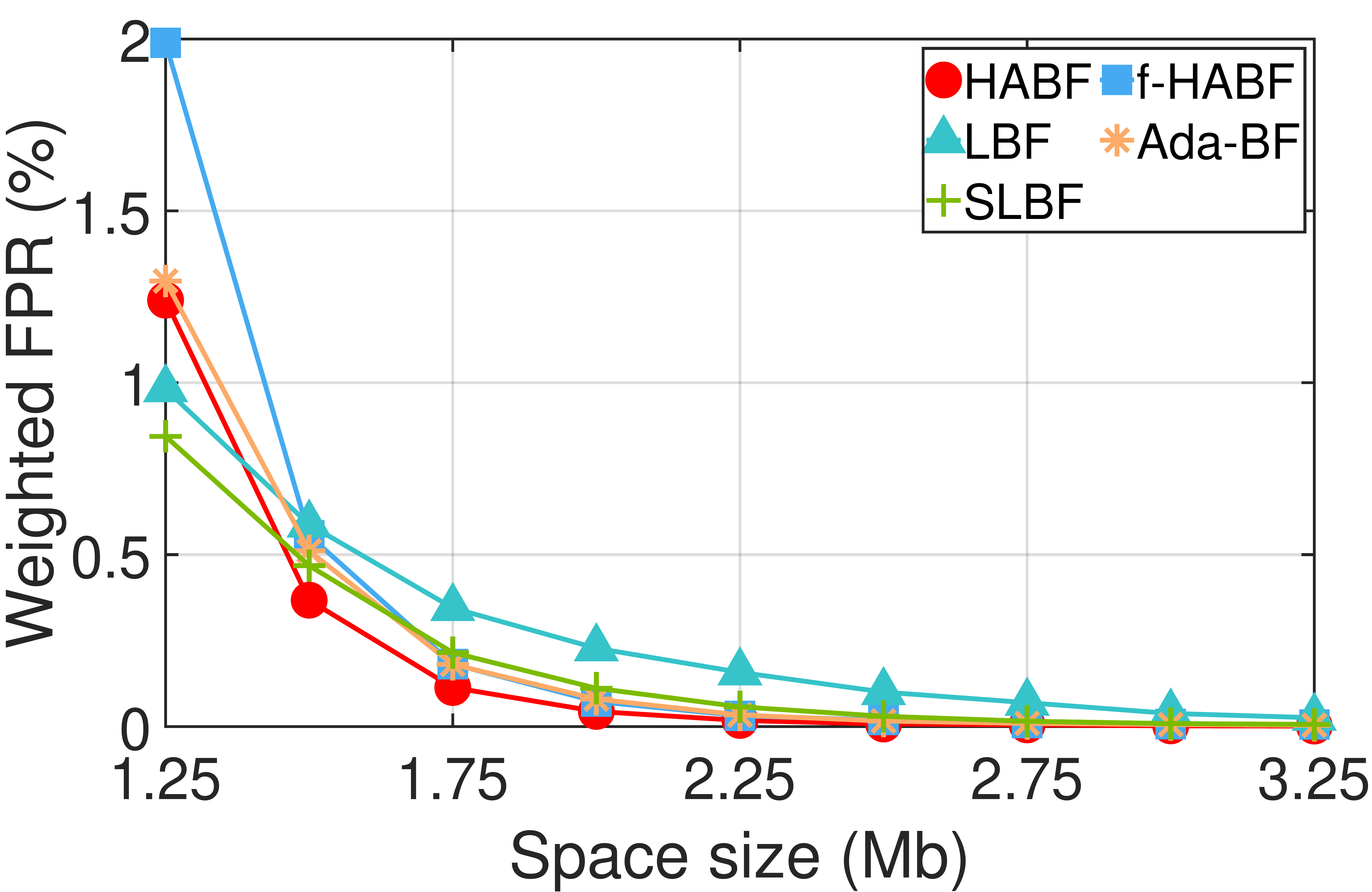}
%\caption{fig2}
\end{minipage}%
}%
\subfigure[vs. Non-learned filter (YCSB)]{
\begin{minipage}[t]{0.24\linewidth}
\centering
\label{fig:fig:11(c)}
\includegraphics[width=1\linewidth]{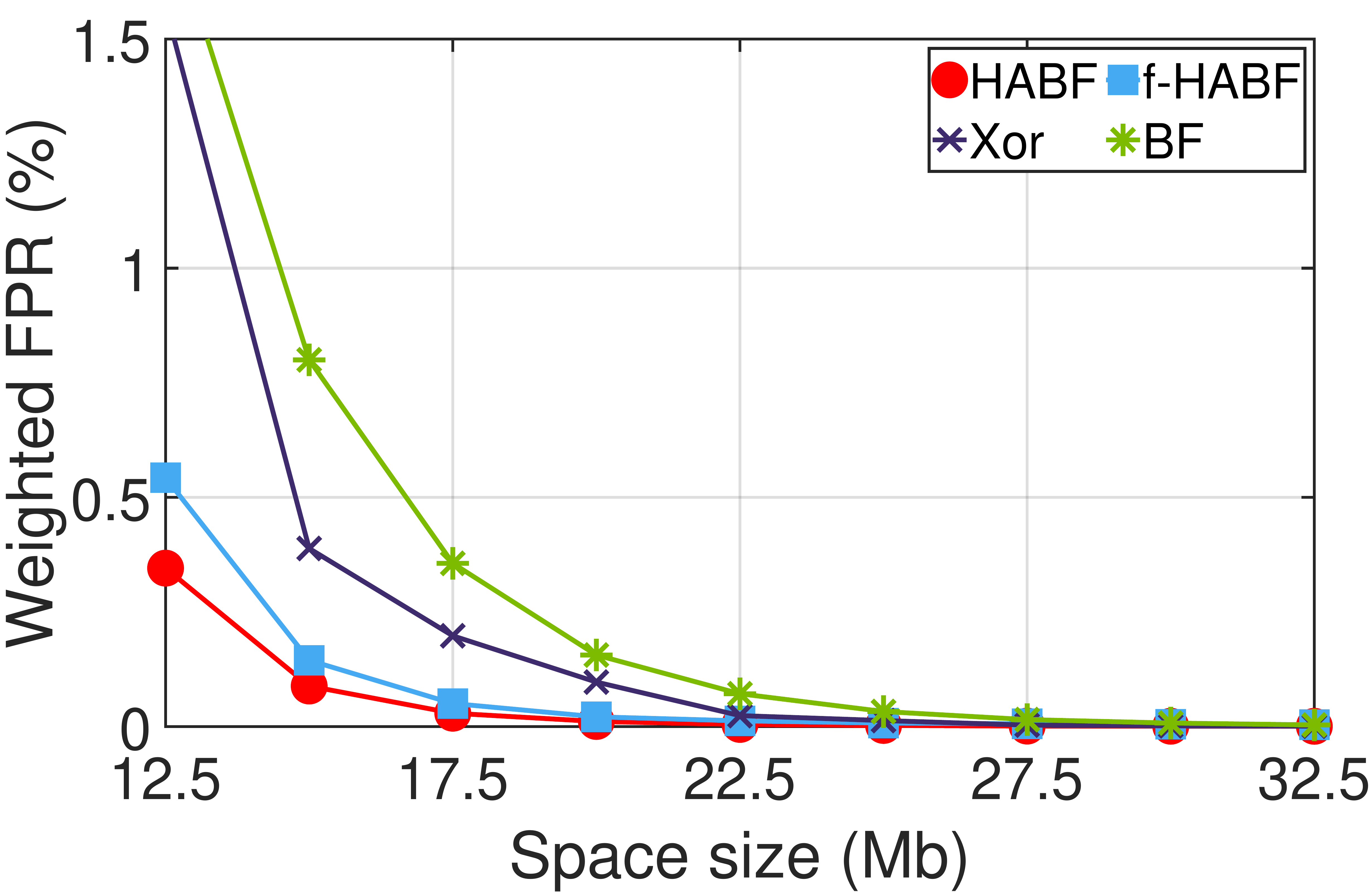}
%\caption{fig2}
\end{minipage}
}%
\subfigure[vs. Learned filter (YCSB)]{
\begin{minipage}[t]{0.24\linewidth}
\centering
\label{fig:fig:11(d)}
\includegraphics[width=1\linewidth]{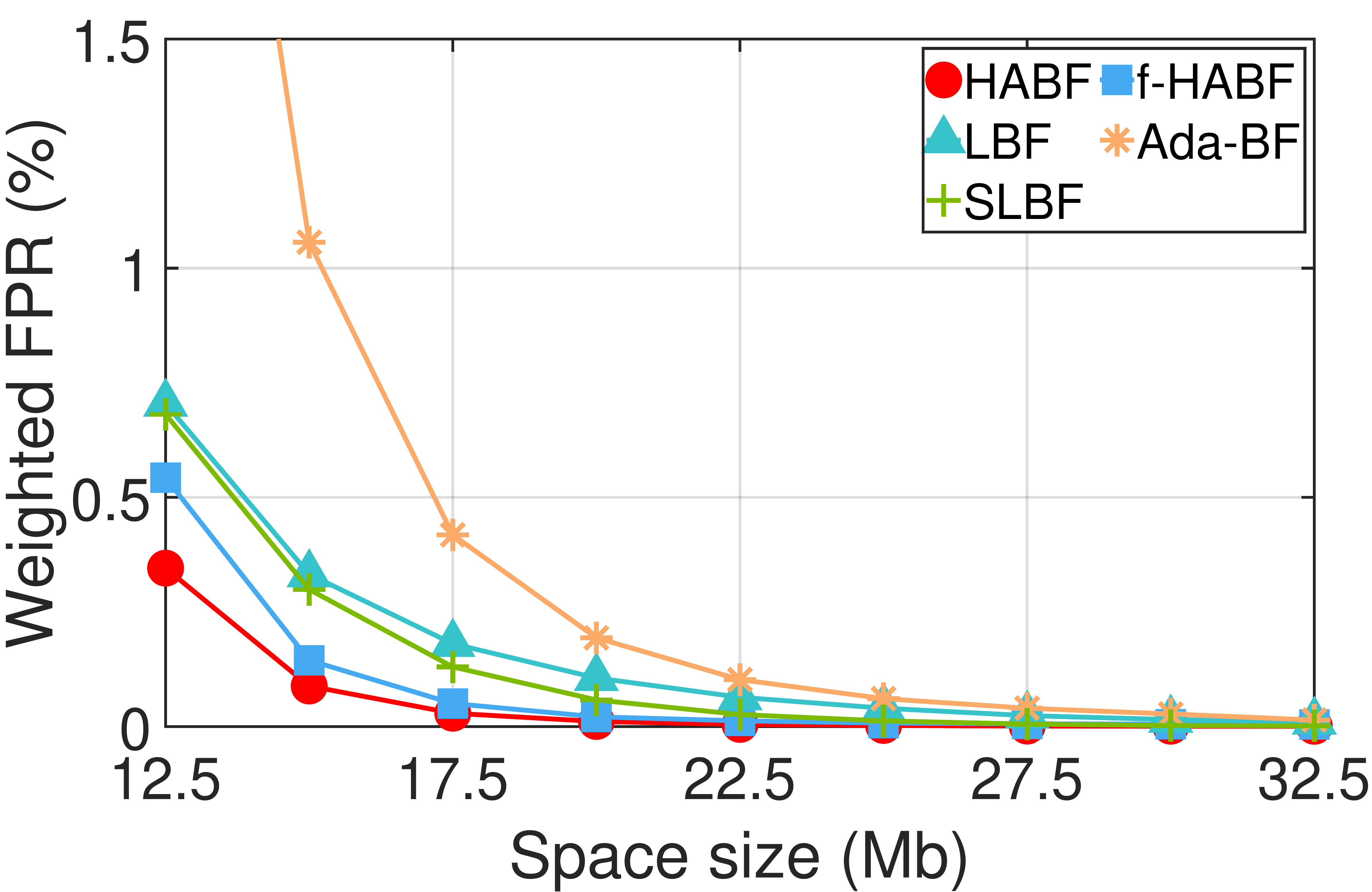}
%\caption{fig2}
\end{minipage}
}%
\centering
\vspace{-4mm}
\label{fig:fig11}
\caption{Weighted FPR on uniform distribution}
\vspace{-7mm}
\end{figure*}
\begin{figure*}[t]
\centering
\subfigure[vs. Non-learned filter (Shalla)]{
\begin{minipage}[t]{0.24\linewidth}
\centering
\label{fig:fig:12(a)}
\includegraphics[width=1\linewidth]{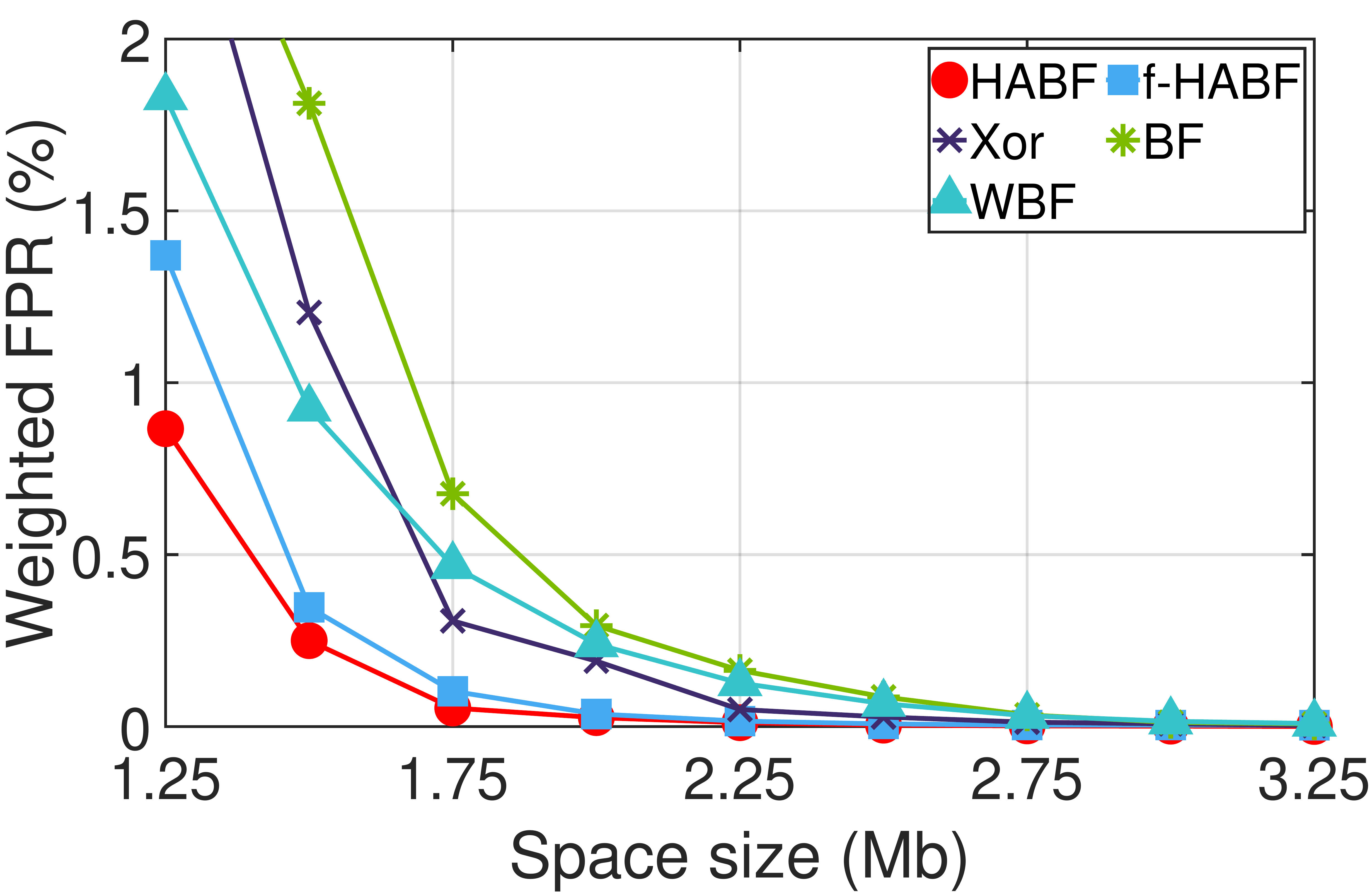}
%\caption{fig1}
\end{minipage}%
}%
\subfigure[vs. Learned filter (Shalla)]{
\begin{minipage}[t]{0.24\linewidth}
\centering
\label{fig:fig:12(b)}
\includegraphics[width=1\linewidth]{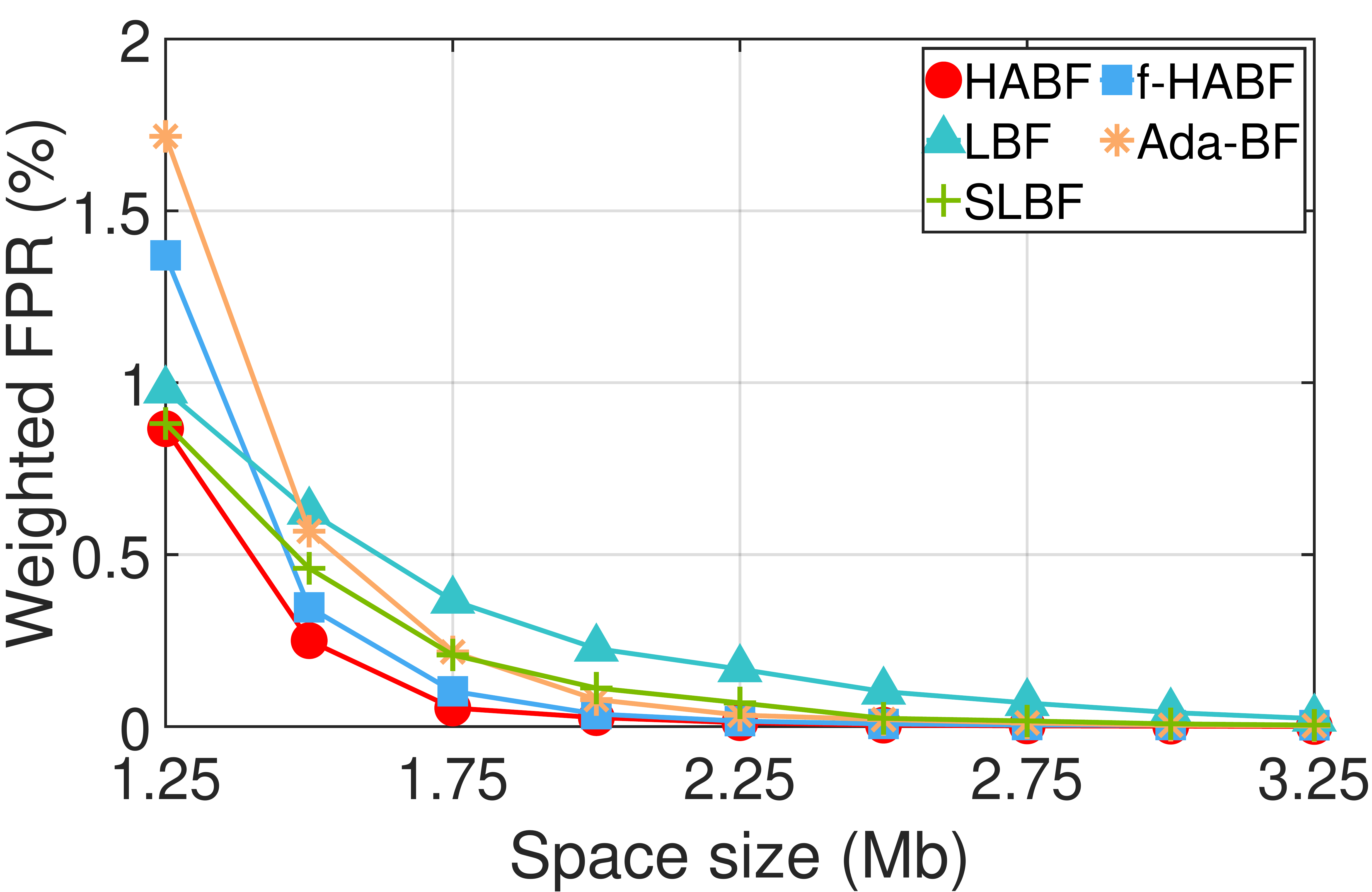}
%\caption{fig2}
\end{minipage}%
}%
\subfigure[vs. Non-learned filter (YCSB)]{
\begin{minipage}[t]{0.24\linewidth}
\centering
\label{fig:fig:12(c)}
\includegraphics[width=1\linewidth]{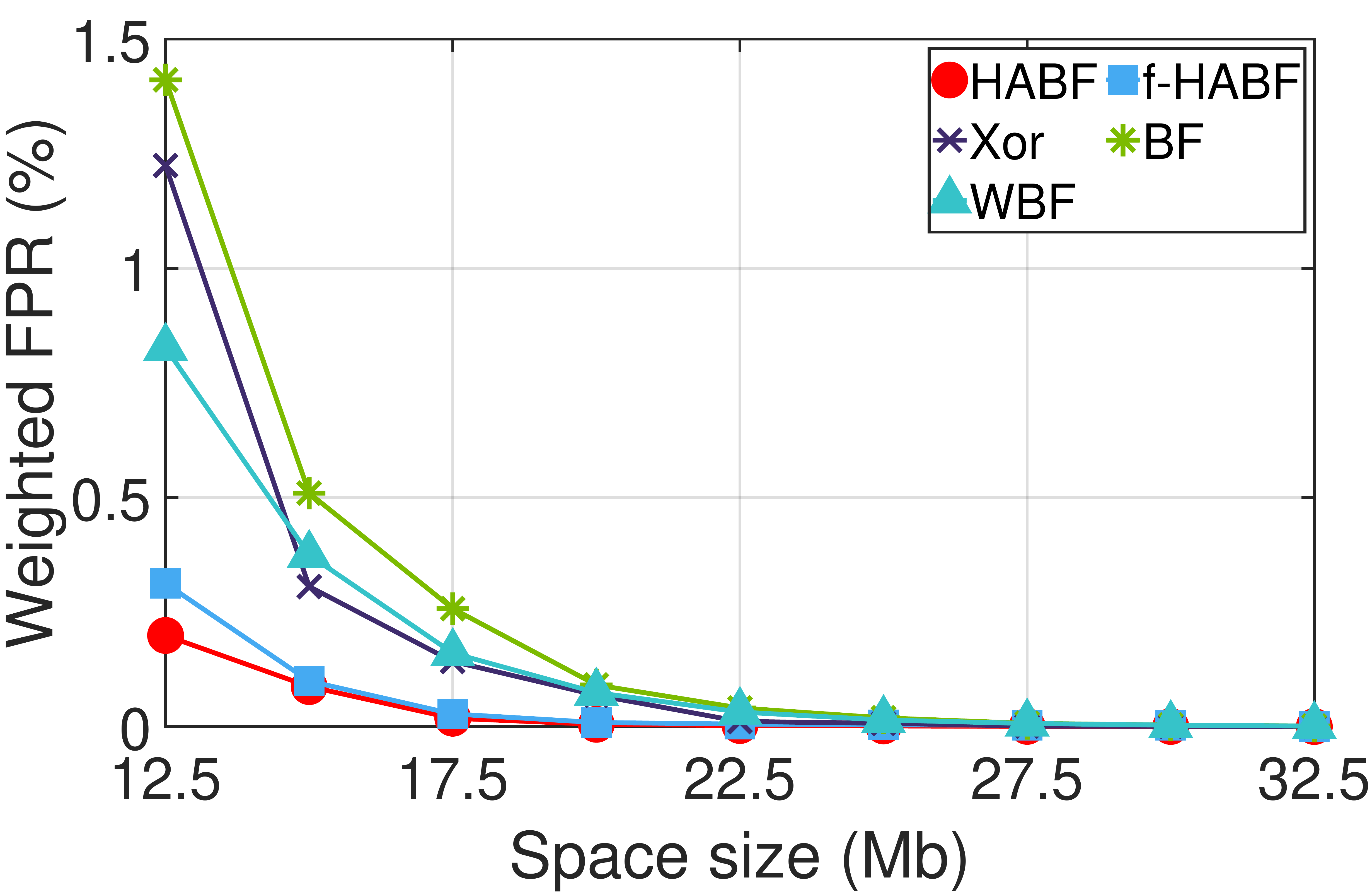}
%\caption{fig2}
\end{minipage}
}%
\subfigure[vs. Learned filter (YCSB))]{
\begin{minipage}[t]{0.24\linewidth}
\centering
\label{fig:fig:12(d)}
\includegraphics[width=1\linewidth]{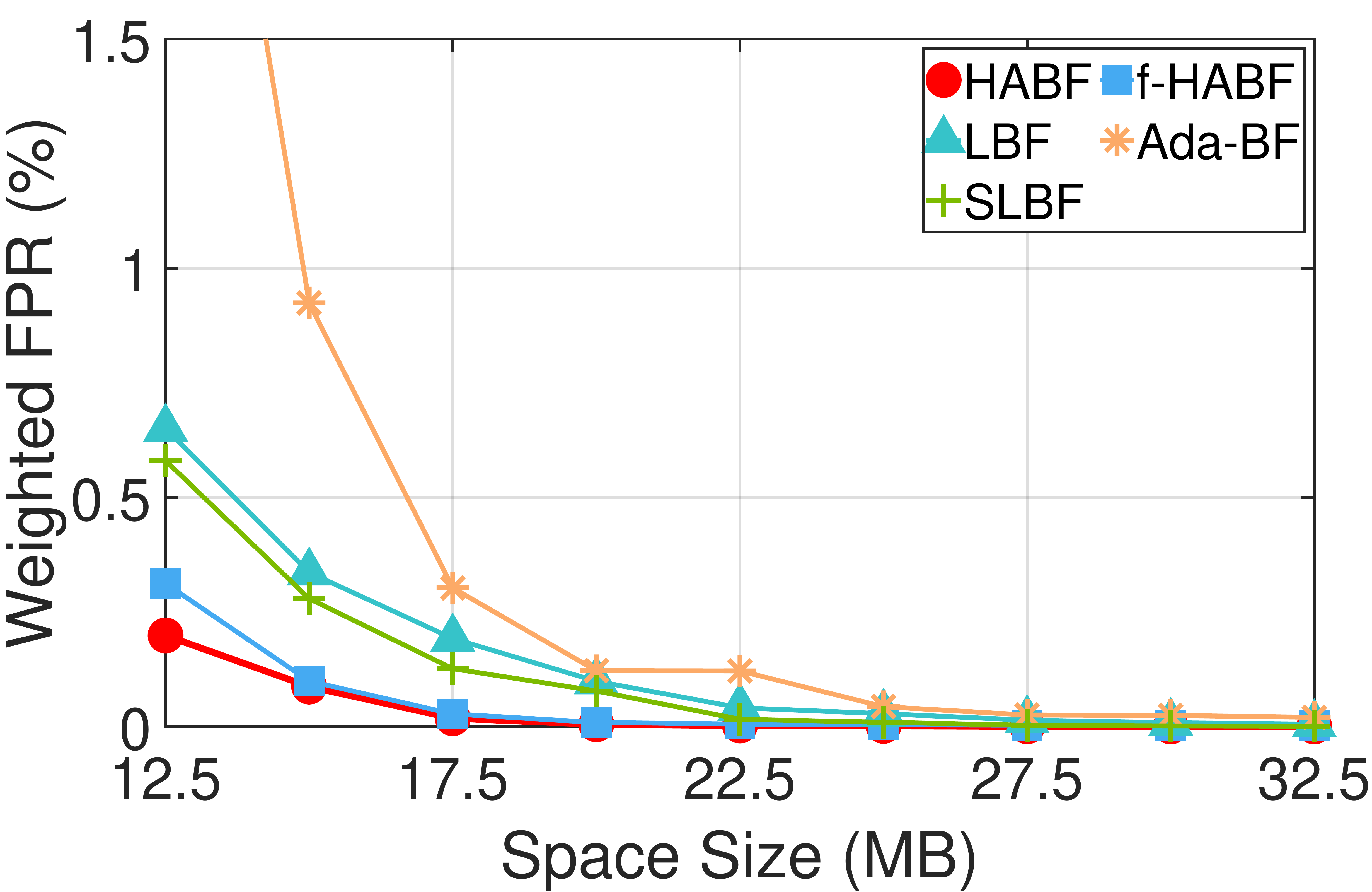}
%\caption{fig2}
\end{minipage}
}%
\centering
\vspace{-4mm}
\label{fig:fig12}
\caption{Weighted FPR on skewed distribution}
\vspace{-7mm}
\end{figure*}
\begin{figure*}[t]
\centering
\subfigure[Construction time (Shalla)]{
\begin{minipage}[t]{0.24\linewidth}
\centering
\label{fig:fig:13(a)}
\includegraphics[width=1\linewidth]{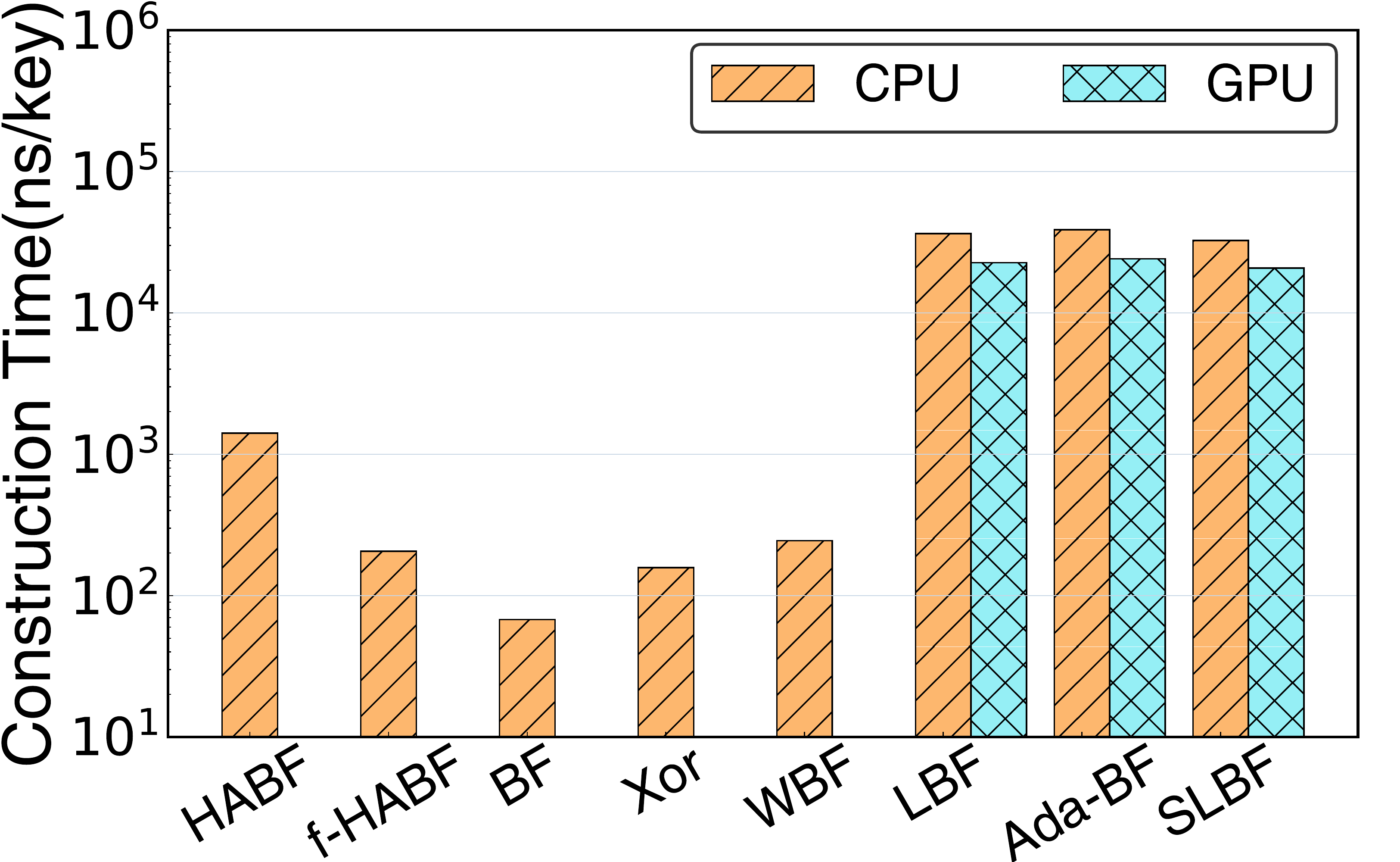}
%\caption{fig1}
\end{minipage}%
}%
\subfigure[Construction time (YCSB)]{
\begin{minipage}[t]{0.24\linewidth}
\centering
\label{fig:fig:13(b)}
\includegraphics[width=1\linewidth]{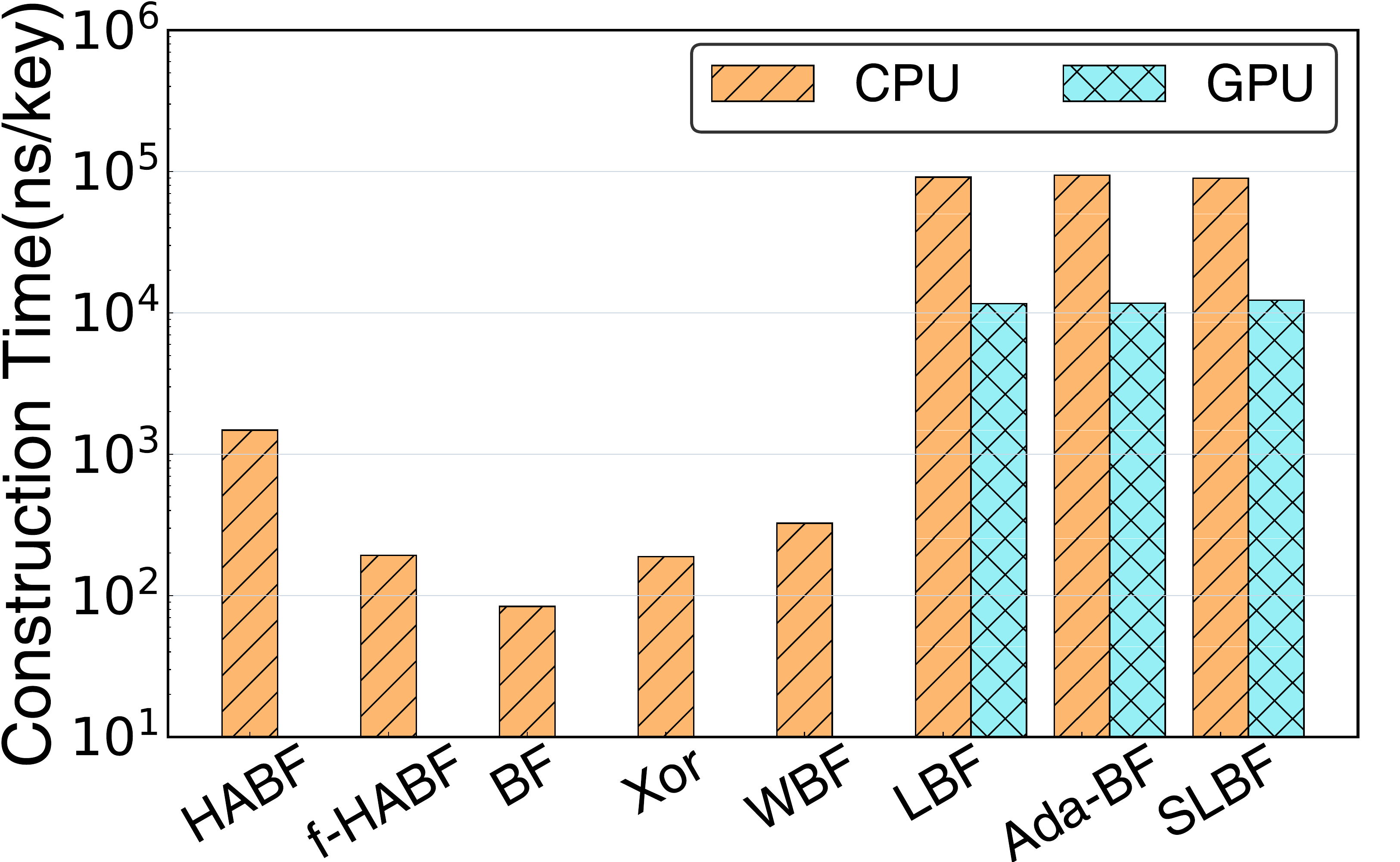}
%\caption{fig2}
\end{minipage}%
}%
\subfigure[Query time (Shalla)]{
\begin{minipage}[t]{0.24\linewidth}
\centering
\label{fig:fig:13(c)}
\includegraphics[width=1\linewidth]{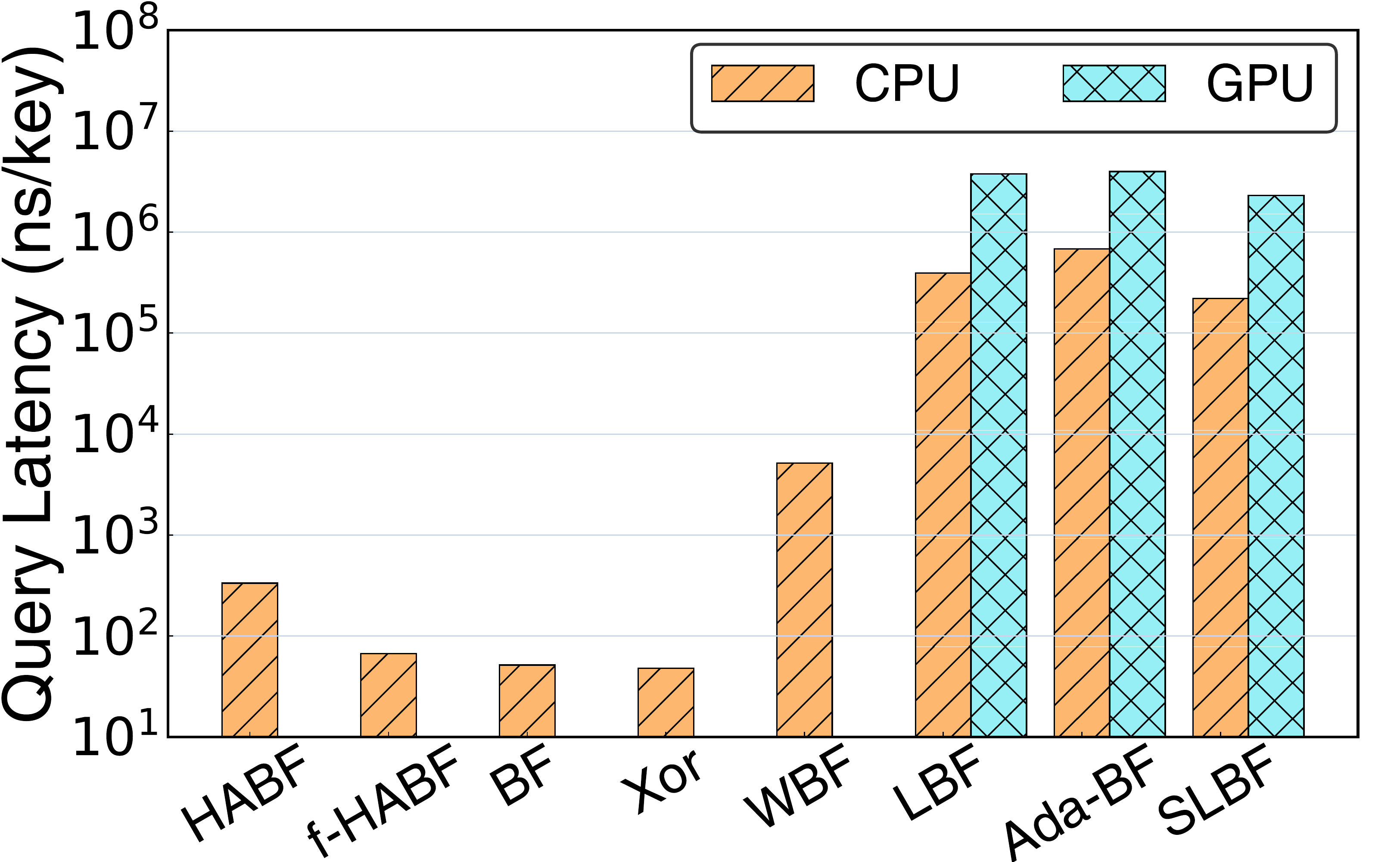}
%\caption{fig2}
\end{minipage}
}%
\subfigure[Query time (YCSB)]{
\begin{minipage}[t]{0.24\linewidth}
\centering
\label{fig:fig:13(d)}
\includegraphics[width=1\linewidth]{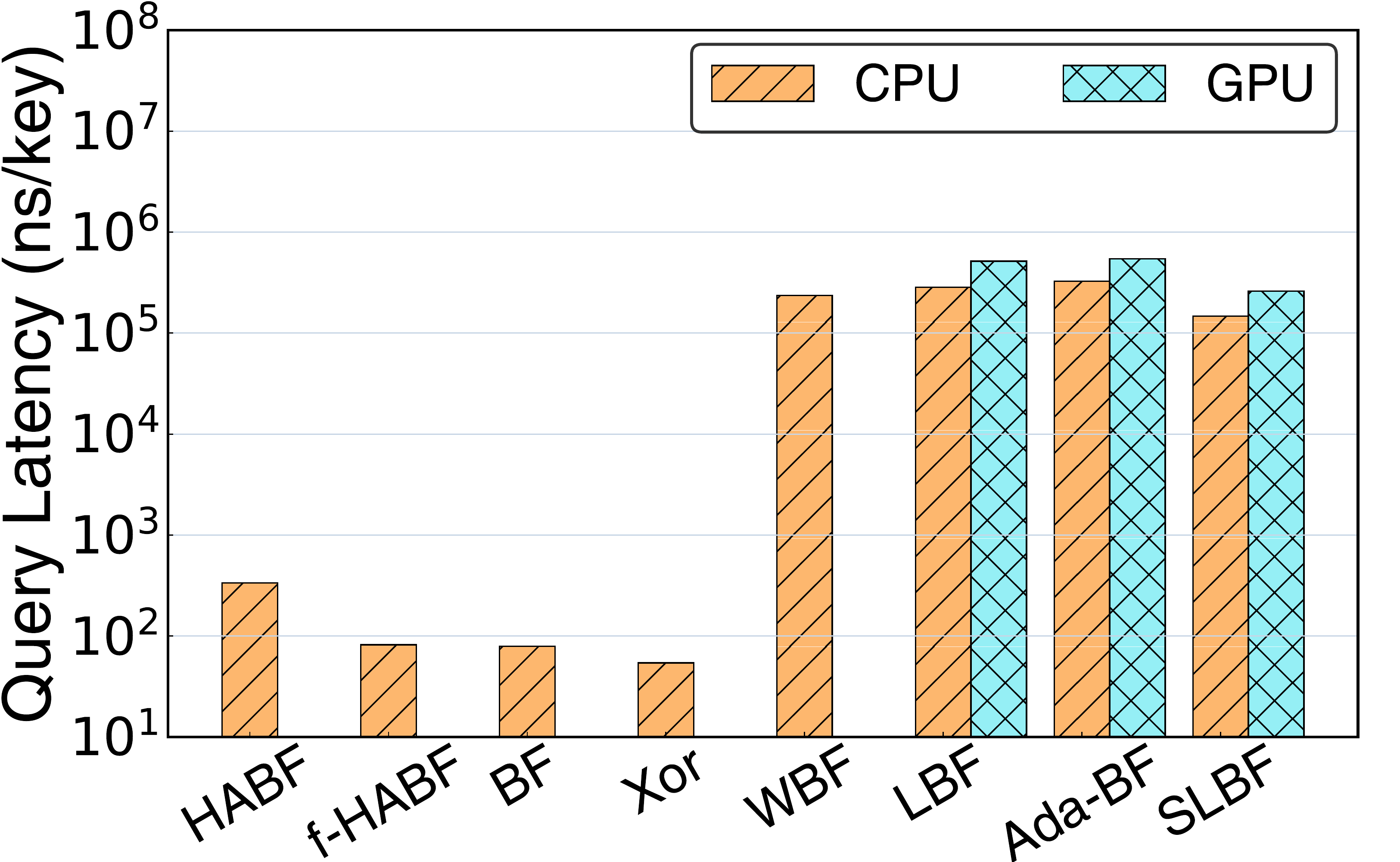}
%\caption{fig2}
\end{minipage}
}%
\centering
\label{fig:fig13}
\vspace{-4mm}
\caption{Construction and query time}
\vspace{-11mm}
\end{figure*}

\subsection{Weighted FPR vs. Space Under Uniform Distribution}
%\blfootnote{\red{$*$ We have optimized the implementation of our algorithm, and the compared results and parameters changes are provided in the appendix \cite{ourSourceCode}.}}
In this experiment, we set the cost distribution of datasets to be uniform.
According to the definition of weighted FPR in Equation (\ref{euqal:wFPR}), the value of cost for each key is normalized to $1$.
We compare the weighted FPRs of HABF and f-HABF with that of BF, Xor, LBF, Ada-BF, SLBF, and WBF.

1) \emph{When the key schema has evident characteristics, HABF will use less space if a low weighted FPR is required. }
For Shalla, we vary the space size from $1.25$MB to $3.25$MB.
As shown in Fig. \ref{fig:fig:11(a)}, HABF always outperforms the non-learned filters regarding weighted FPR with the same space size.
As shown in Fig. \ref{fig:fig:11(b)}, since the URL blacklist has evident characteristics, learned filters can use only a small space to correctly judge a large part of the keys.
At this time, learned filters will consume a small space to achieve the same weighted FPR.
But with lower requirements for weighted FPR, learned filters need more space than HABF.
When increasing the space size to $1.5$MB, the weighted FPR of BF, Xor, LBF, Ada-BF and SLBF is $1.73\%$, $1.56\%$, $0.54\%$, $0.51\%$, and $0.44\%$, respectively, while HABF achieves $0.36\%$ and f-HABF achieves $0.55\%$.

2) \emph{ When the key schema is approximately random, HABF has the smallest weighted FPR for all our space settings.} For YCSB, we vary the space size from $12.5$MB to $32.5$MB.
As shown in Fig. \ref{fig:fig:11(c)} and Fig. \ref{fig:fig:11(d)}, the weighted FPR of HABF changes from $3.46\times10^{-3}$ to $3.63\times 10^{-6}$, and the weighted FPR of f-HABF is around $1.5\times$ than HABF on average, while the weighted FPR of BF, Xor, LBF, Ada-BF, and SLBF change from $1.78\times10^{-2}$ to $2.83\times 10^{-5}$, $1.57\times10^{-2}$ to $1.59\times 10^{-5}$, $7.04\times10^{-3}$ to $1.08\times 10^{-4}$, $3.13\times10^{-2}$ to $1.42\times 10^{-4}$, and $6.81\times10^{-3}$ to $1.72\times 10^{-5}$, respectively.
The randomness of the key schema characteristics increases the difficulty of fitting ML model, and the performance of algorithms relying on the prediction score \cite{2019adaptive} of the ML model will be greatly affected.
%
%ML model in LBF classifies a key to be positive when the score is greater than a preset threshold.
%
There is a significant gap in performance between the two datasets for Ada-BF.
%
%Ada-BF builds Bloom filters to test a key with different numbers of hash functions based on the score, which is highly dependent on ML model.
%
By adding a Bloom filter in the beginning to reduce the impact of ML model errors, the performance of SLBF will be less affected.
\begin{figure*}[t]
\end{figure*}
\subsection{Weighted FPR vs. Space Under Skewed Distribution}
In this experiment, we vary space size as in the previous experiment and set the cost distribution of datasets to be Zipf with skewness $1.0$.
Consequently, the weighted FPR is mostly contributed by the false positives of keys with high cost.

\emph{HABF always has the smallest weighted FPR under all the space settings.}
For Shalla, as shown in Fig. \ref{fig:fig:12(a)}, compared with non-learned filter, the weighted FPR of HABF changes from $8.67\times10^{-3}$ to $2.56\times 10^{-6}$ and the weighted FPR of f-HABF changes from $1.37\times10^{-2}$ to $3.86\times 10^{-6}$, while the weighted FPR of BF, Xor, and WBF change from $2.81\times10^{-2}$ to $7.49\times 10^{-5}$, $2.67\times10^{-2}$ to $2.74\times 10^{-5}$, and $1.83\times10^{-2}$ to $8.81\times 10^{-5}$, respectively.
As shown in Fig. \ref{fig:fig:12(b)}, compared with learned filter, the weighted FPR of LBF, Ada-BF, and SLBF change from $9.78\times10^{-3}$ to $2.3\times 10^{-4}$, $1.72\times10^{-2}$ to $2.13\times 10^{-5}$, and $8.81\times10^{-3}$ to $4.05\times 10^{-5}$, respectively.
It shows that HABF performs better under the skewed cost distribution.
For YCSB, as shown in Fig. \ref{fig:fig:12(c)} and Fig. \ref{fig:fig:12(d)}, compared with both non-learned and learned filters, the weighted FPR of HABF reaches the range from $1.99\times10^{-3}$ to $1.97\times 10^{-6}$. While for other algorithms, the lowest weighted FPR changes from $5.80\times10^{-3}$ to $5.14\times 10^{-6}$.
\vspace{-3mm}

\subsection{Effect of Skewness}

\begin{wrapfigure}{l}{0.23\textwidth}
  \vspace{-15pt}
  \begin{center}
   \includegraphics[width=0.25\textwidth]{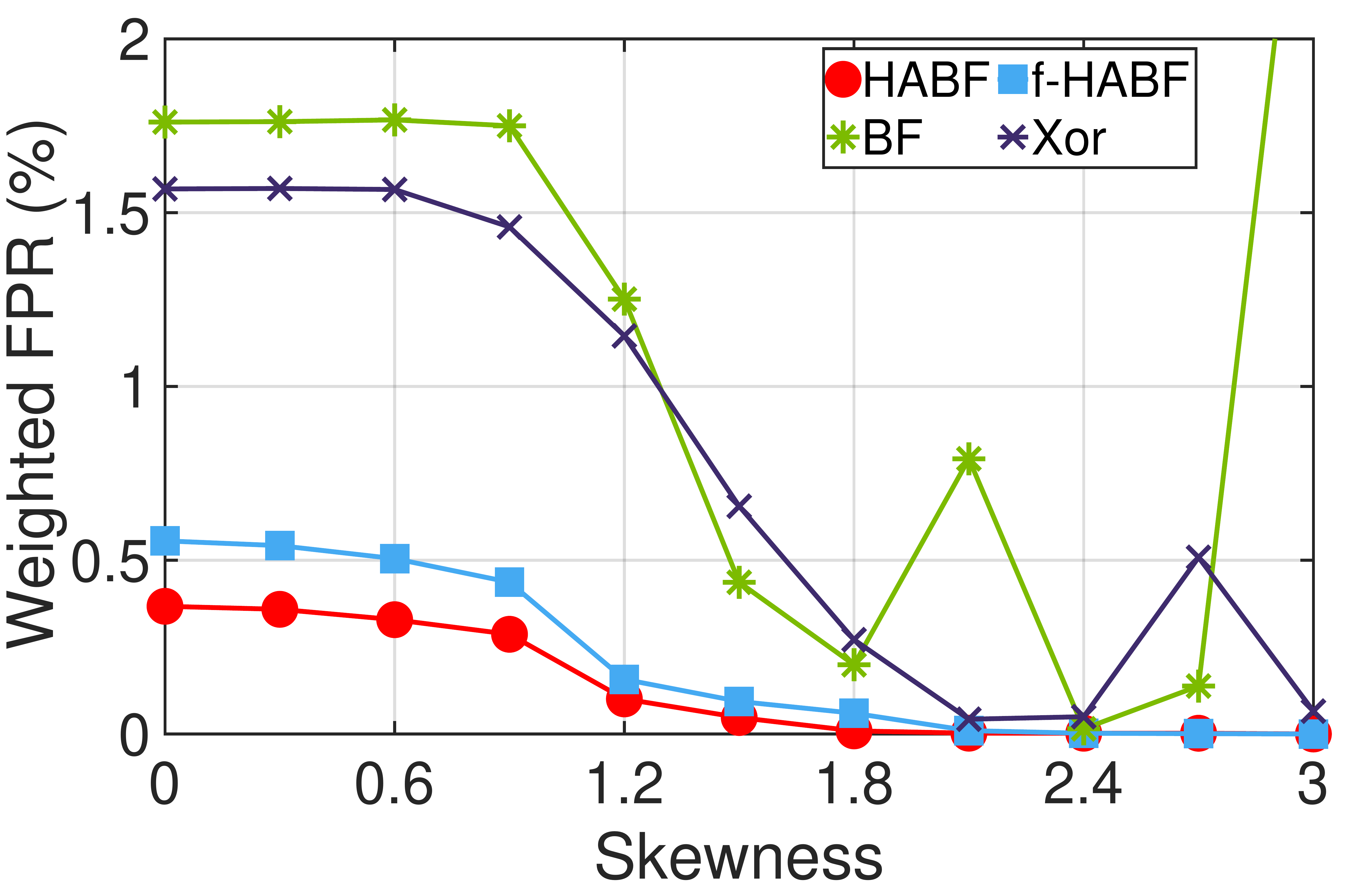}
  \end{center}
  \vspace{-15pt}
  \caption{Varying skewness}
    \label{fig:skewness}
  \vspace{-18pt}
\end{wrapfigure}
We further study how the skewness of dataset affects the weighted FPR as shown in Fig. \ref{fig:skewness}. Here we use Shalla dataset and set the space size to $1.5$MB, and we show how the weighted FPR changes as the skewness increases from $0.0$ to $3.0$ for HABF, f-HABF, BF, and Xor.
When the skewness is $0$, the weighted FPRs follows Fig. \ref{fig:fig:11(a)}.
When the skewness $\geq 0.9$, the weighted FPRs of HABF and f-HABF continue to decrease steadily but for BF and Xor, the weighted FPRs show great fluctuations.
The reason is that, as the skewness increases, once a key with high cost is misjudged, the weighted FPR increases a lot.
Therefore, BF and Xor hardly show any performance gain as they are insensitive to the cost distribution.
\begin{figure*}[t]
\begin{minipage}[t]{0.49\textwidth}
\subfigure[On uniform distribution (YCSB)]{
\begin{minipage}[t]{0.49\linewidth}
\centering
\label{fig:BFcompare(a)}
\includegraphics[width=1\linewidth]{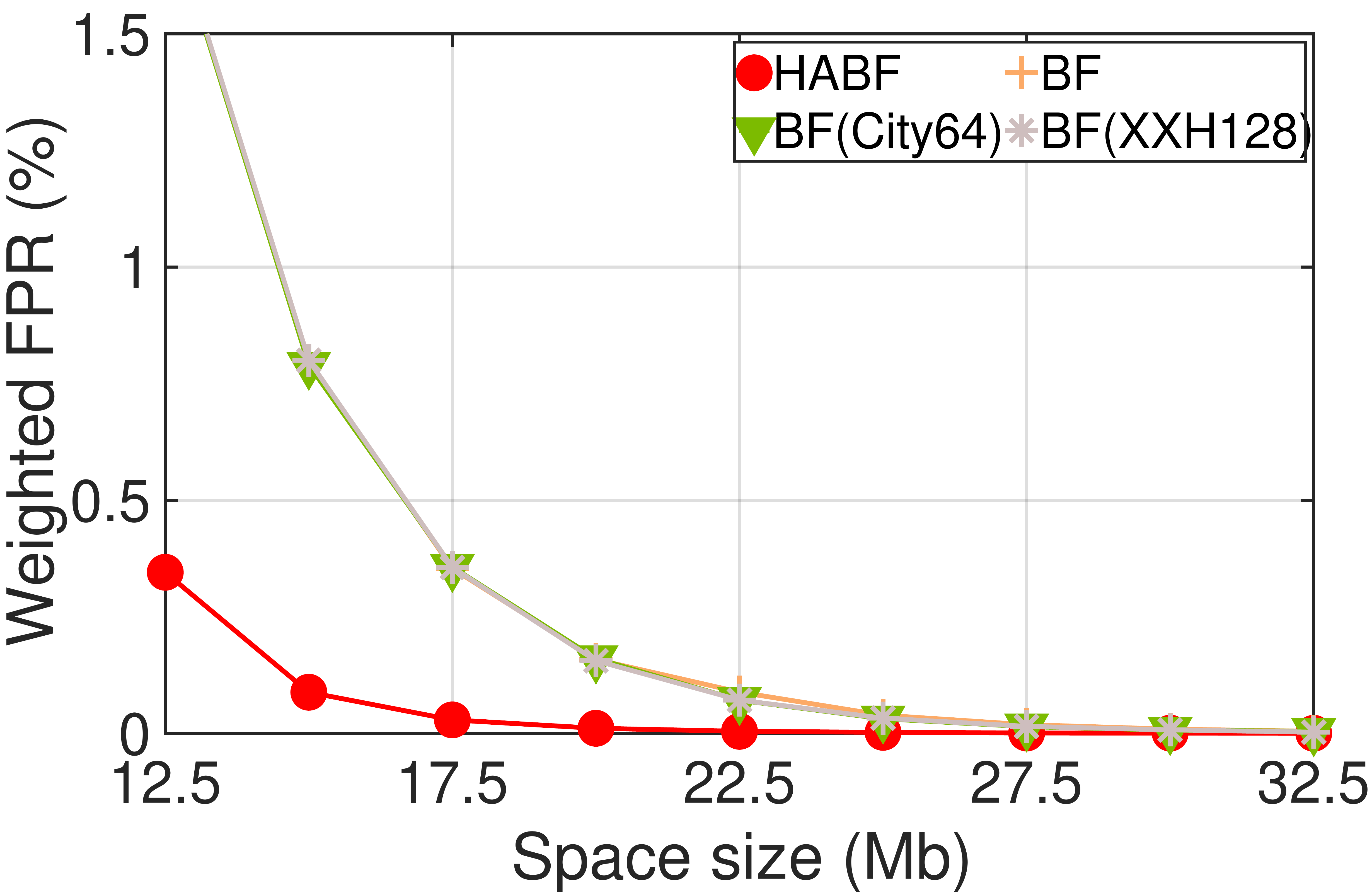}
\end{minipage}
}%
\subfigure[On skew distribution (YCSB)]{
\begin{minipage}[t]{0.49\linewidth}
\centering
\label{fig:BFcompare(b)}
\includegraphics[width=1\linewidth]{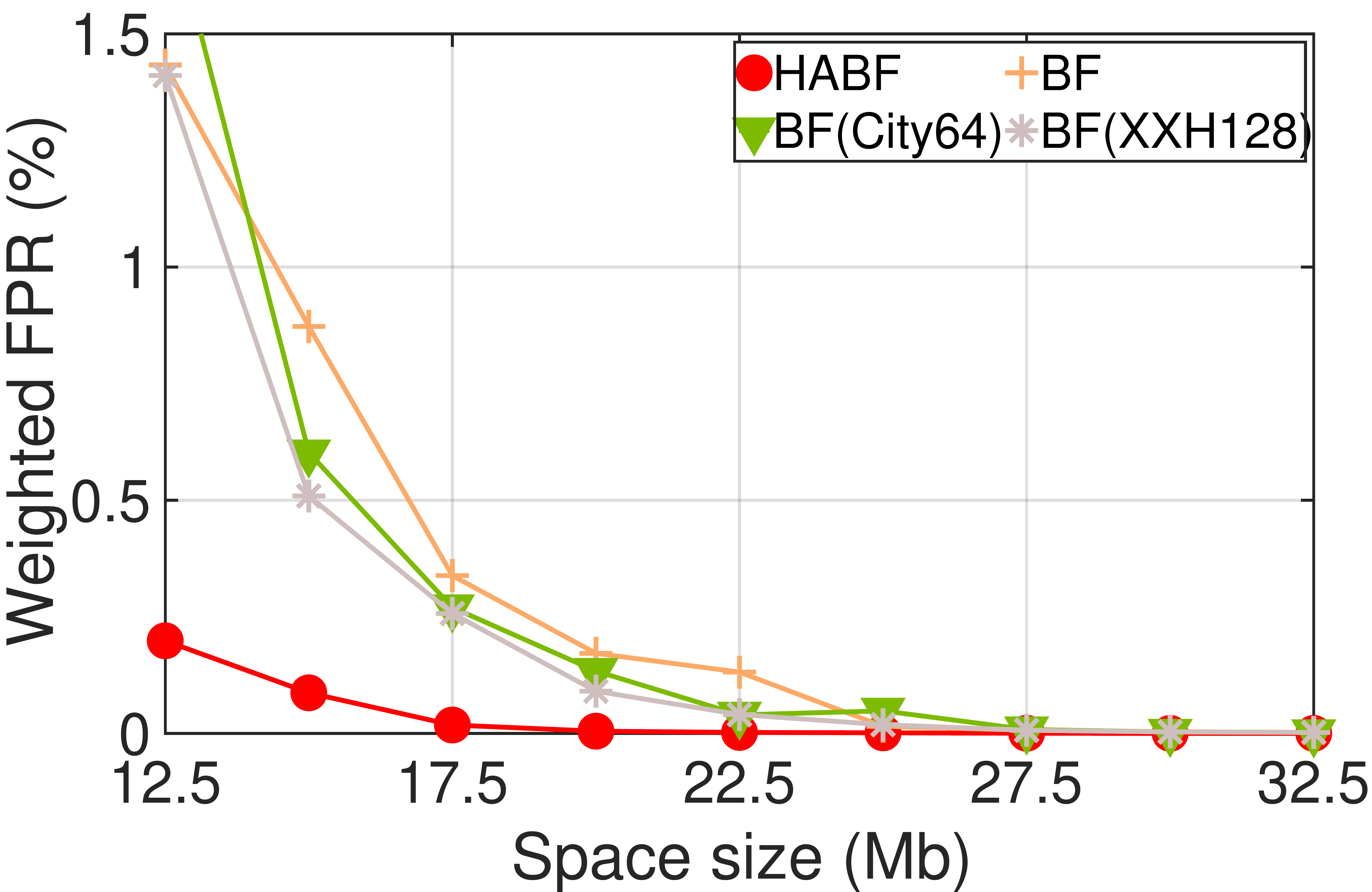}
\end{minipage}
}%
\vspace{-3mm}
\caption{Bloom filter with different implementations}
\label{fig:BFcompare}
\end{minipage}
\begin{minipage}[t]{0.49\textwidth}
\subfigure[Shalla]{
\begin{minipage}[t]{0.49\linewidth}
\centering
\label{fig:fig:14(a)}
\includegraphics[width=1\linewidth]{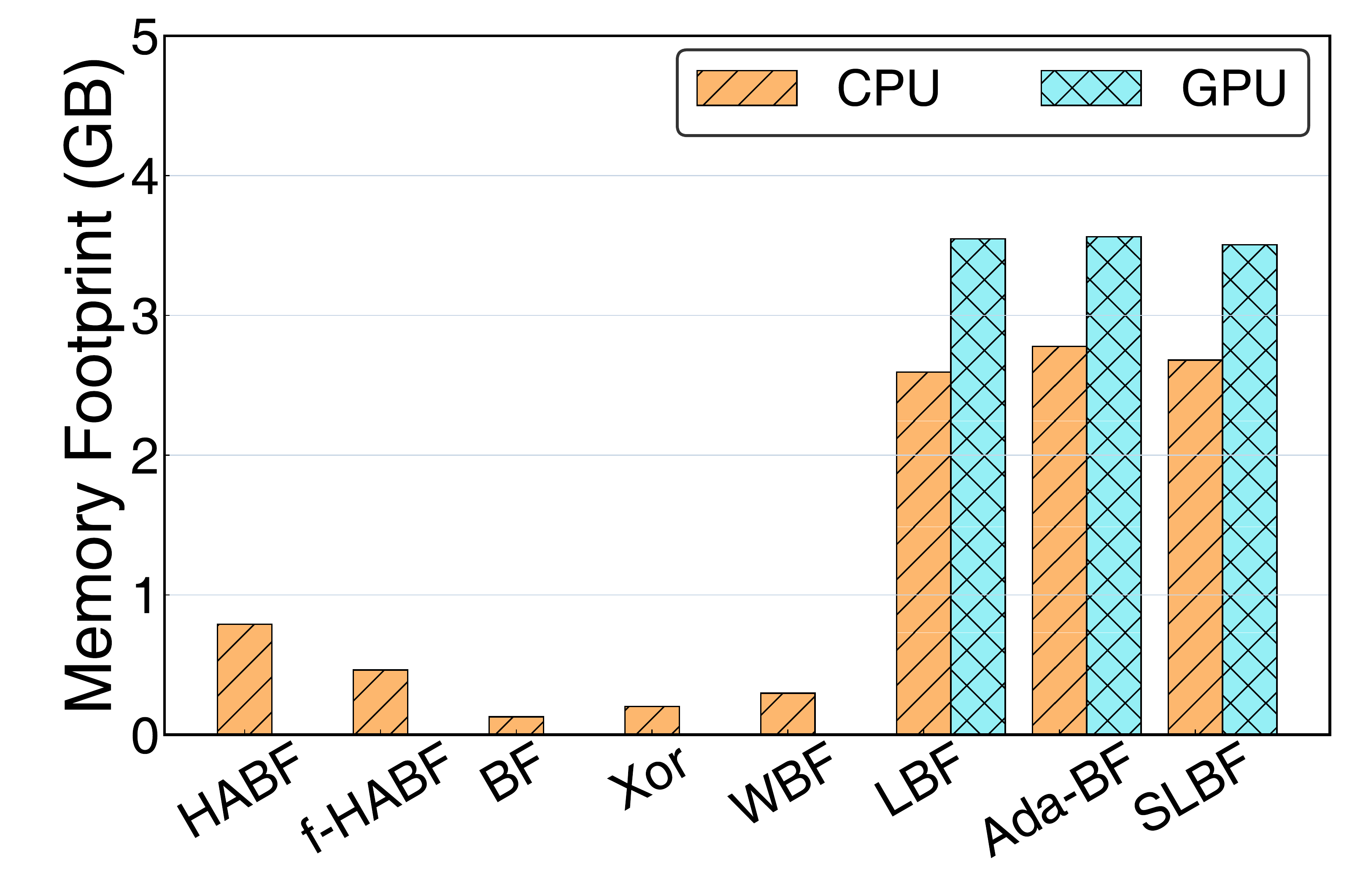}
\end{minipage}
}%
\subfigure[YCSB]{
\begin{minipage}[t]{0.49\linewidth}
\centering
\label{fig:fig:14(b)}
\includegraphics[width=1\linewidth]{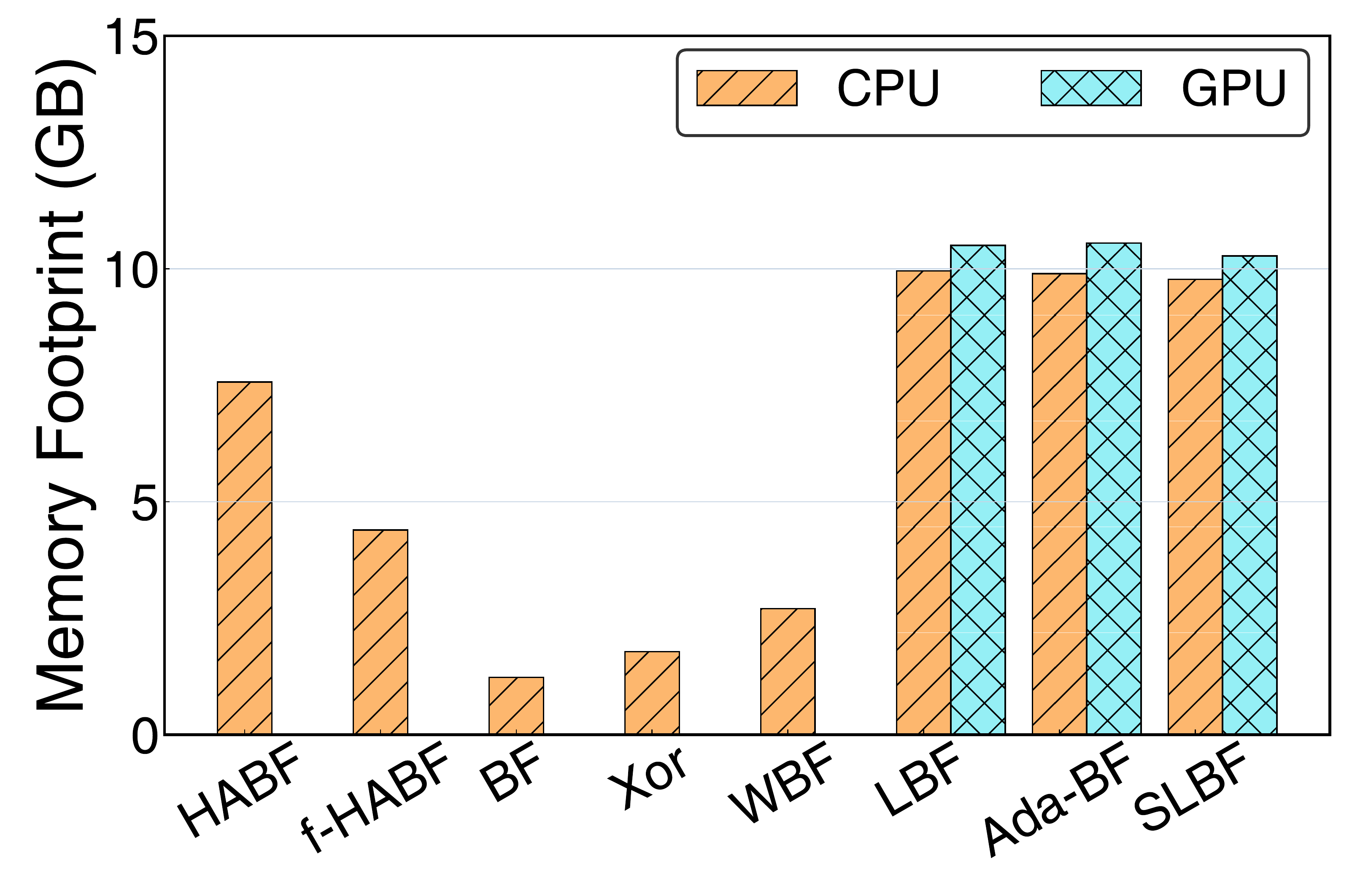}
\end{minipage}
}%
\vspace{-3mm}
\caption{Memory footprint of construction}
\end{minipage}
\vspace{-8mm}
\end{figure*}

\subsection{Discussion for Bloom filter with different implementations}
Since the performance of the Bloom filter will be affected by different implementations of the hash function.
We implement three versions of Bloom filter: BF by using $k$ different hash functions in Table~\ref{Tab:hashfunc}.
BF (City64) by using CityHash (64bit version) and BF (XXH128) by using xxHash (128bit version).
For the latter two implementations, we use different seeds to generate $k$ hash values for accuracy.
As shown in Fig.~\ref{fig:BFcompare}, the dataset is set to YCSB, under the uniform distribution, since the cost of each key is the same, the three versions of the Bloom filter ({\it i.e.} BF, BF (City64) and BF (XXH128)) are nearly consistent.
Under skewed distribution, we set the cost distribution to skewness $1.0$, all Bloom filter implementations have fluctuated.
It demonstrates that even with advanced hash functions like City64 and XXH128, they still can't effectively reduce weighted FPR and are not sensitive to the skew cost distribution.
\subsection{Construction and Query Time}
In this part, we compare the construction time and the query time in nanoseconds per key. We fix the space size for each algorithm, \ie, $1.5$MB for Shalla and $15$MB for YCSB.

1) \emph{The construction time per key of HABF and f-HABF are around $19\times$ and $2.7\times$ larger than that of BF, respectively.}
On the Shalla dataset, as shown in Fig. \ref{fig:fig:13(a)}, for HABF, the construction time per key is $1$,$411ns$; for f-HABF, it is around $205ns$; for BF, it is around $68ns$; for Xor, it is around $158ns$; for WBF, it is around $245ns$; while for learned filters, CPU-based LBF, Ada-BF, and SLBF are around $36$,$430ns$, $38$,$743ns$, and $32$,$470ns$, respectively.
LBF (GPU), Ada-BF (GPU), and SLBF (GPU) take $25$,$686 ns$, $24$,$123 ns$, and $20$,$728 ns$, respectively.
On the YCSB dataset, as shown in Fig. \ref{fig:fig:13(b)}, HABF, f-HABF, BF, Xor ,and WBF take $1$,$480 ns$, $193 ns$, $84 ns$, $188 ns$, and $325 ns$, respectively.
LBF (GPU), Ada-BF (GPU), and SLBF (GPU) take $11$,$636 ns$, $11$,$730 ns$, and $12$,$300 ns$, respectively, while the construction time of CPU-based learning models is all above $90$,$000 ns$.
The construction of the learning model is highly dependent on GPU especially for massive data, but for the machines without GPU, the application of learning models is heavily limited.
Our fast version, \ie, f-HABF achieves the same order of construction speed as BF and Xor.
%
%However, the construction time of HABF can be faster than learning model using GPU and achieve higher optimization performance.
%\vspace{-6mm}

2) \emph{The query time of HABF and f-HABF per key is around $5.35\times$ and $1.15\times$ than that of BF, respectively.}
Fig. \ref{fig:fig:13(c)} and \ref{fig:fig:13(d)} show the average query time of all algorithms in two datasets.
For Shalla, to query one key, HABF, f-HABF, BF, and Xor take $338 ns$, $67 ns$, $52ns$, and $48ns$, respectively. For YCSB, HABF takes $336 ns$; f-HABF, BF, and Xor take $82 ns$, $79 ns$, and $54 ns$, respectively.
This result indicates the potential of the application of HABF in real-time query scenarios.
The query time of LBF, Ada-BF, and SLBF are all above $500\times$ larger than that of HABF due to computational complexity of ML model, and using GPU to query a key may increase the query time due to the transmission of data between CPU and GPU.
For WBF, it will traverse the cached cost list when querying a key, which shows that WBF will lead to poor query performance with the size of the cost list increasing.
\subsection{Construction Memory Consumption}
In this part, we fix the space size of each algorithm \ie, $1.5$MB for Shalla and $15$MB for YCSB, and compare the CPU memory footprint during construction.
Moreover, we also give CPU memory usage for some algorithms using GPU; note that here we allocate all memory of two GPUs to these algorithms.

\emph{The construction memory consumption of HABF and f-HABF is around $6.1\times$ and $3.6\times$ than that of BF, which is lower than all learning models.}
For Shalla, as shown in Fig. \ref{fig:fig:14(a)}, HABF, f-HABF, BF, Xor and WBF consume $0.79$GB, $0.46$GB, $0.13$GB, $0.20$GB and $0.58$GB, respectively.
LBF, Ada-BF, and SLBF consume $2.59$GB, $2.78$GB, and $2.68$GB, respectively.
Due to the process of dealing and loading data to GPU, LBF (GPU), Ada-BF (GPU), and SLBF (GPU) consume more CPU memory which is $3.55$GB, $3.56$GB, and $3.51$GB, respectively.
For YCSB, as shown in Fig. \ref{fig:fig:14(b)}, HABF, f-HABF, BF, Xor and WBF consumes $7.569$GB, $4.394$GB, $1.23$GB, $1.781$GB, and $2.708$GB, respectively.
CPU-based and GPU-based learning models consume memory of $9.88$GB and $10.44$GB on average, respectively.
The reason for extra memory for HABF during construction is that HABF needs to maintain negative keys and two runtime auxiliary data structures.

%% file: Latex/conclusion.tex
%\vspace{0mm}

\section{Conclusion}
\label{sec:conclusion}

In this paper, we study the problem of how to customize the hash functions for positive keys to minimize the overall cost of the misidentified negative keys when the information of negative keys and their costs are available.
%
%\red{
%Especially, we consider the following scenarios: the information of negative keys are available, and further, the misidentification of negative keys brings different cost.
%}
%
We propose a novel framework Hash Adaptive Bloom Filter (HABF), which consists of a standard Bloom filter, and a novel lightweight hash table named HashExpressor for storing the customized hash functions.
Then, at query time, to provide a one-side error guarantee, HABF follows a two-round pattern to check whether a key is in the set.
Besides, to optimize hash function selections for positive keys, a greedy-based but performance-bounded TPJO algorithm is proposed.
%
%Moreover, we theoretically analyze the performance of HABF and give a performance bound of HABF.
%
Extensive experiments show that HABF outperforms the standard Bloom filter and its variants on the whole in terms of accuracy, construction time, query time, and memory space consumption. 

%% file: Latex/acknowledge.tex
\section*{Acknowledgment}
\vspace{-1mm} 
\label{sec:acknowledgment}
We thank the reviewers for their thoughtful suggestions. 
This work was supported in part by the National Natural Science Foundation of China under Grant 61872178, %戴海鹏面上
in part by the Natural Science Foundation of Jiangsu Province under Grant No. BK20181251, %戴海鹏江苏省面上
in part by the open research fund of Key Lab of Broadband Wireless Communication and Sensor Network Technology (Nanjing University of Posts and Telecommunications), Ministry of Education, %戴海鹏南邮项目
in part by the Key Research and Development Project of Jiangsu Province under Grant No. BE2015154 and BE2016120,
in part by the National Natural Science Foundation of China under Grant 61832005, %戴海鹏参与的叶保留老师重点项目
and 61672276, %戴海鹏参与的窦万春老师面上项目
in part by the Collaborative Innovation Center of Novel Software Technology and Industrialization, Nanjing University, %窦老师
in part by the Jiangsu High-level Innovation and Entrepreneurship (Shuangchuang) Program,  %窦老师 
and in part by the National Natural Science Foundation of China (NO.62072230, U1811461) and Alibaba Innovative Research Project. % 李猛
\vspace{-2mm} 